\documentclass[11pt]{article}       
\usepackage{geometry}               
\usepackage{graphicx}
\usepackage{caption}
\usepackage{subcaption}
\usepackage{comment}
\usepackage{cite}
\usepackage{amsmath} % assumes amsmath package installed
\usepackage{amssymb}  % assumes amsmath package installed
\usepackage{amsthm}  % assumes amsmath package installed
\usepackage{bm}
\usepackage{lipsum}
\usepackage{algorithm}
\usepackage{algorithmic}
\usepackage{color}
\usepackage{authblk}
\usepackage{tabularx}
\usepackage{textcomp}

\newcolumntype{C}[1]{>{\centering\arraybackslash}p{#1}}

\newtheorem{theorem}{Theorem}
\newtheorem{lemma}{Lemma}
\newtheorem{remark}{Remark}
\newtheorem{proposition}{Proposition}
\newtheorem{corollary}{Corollary}
\newtheorem{definition}{Definition}
\numberwithin{equation}{section}

\DeclareMathOperator*{\argmax}{arg\,max}
\DeclareMathOperator*{\argmin}{arg\,min}

\title{Multi-Objective Predictive Taxi Dispatch via\\ Network Flow Optimization\thanks{This research was supported in part by  the Creative-Pioneering Researchers Program through SNU, the Basic Research Lab Program through the National Research Foundation of Korea funded by the MSIT(2018R1A4A1059976), and BK21 Plus Project in 2019. Corresponding author: Insoon Yang (e-mail: insoonyang@snu.ac.kr).}}

\author[1]{Beomjun Kim}
\author[2]{Jeongho Kim}
\author[1,3]{Subin Huh}
\author[4]{Seungil You}
\author[1,3]{Insoon Yang}

\affil[1]{Department of Electrical and Computer Engineering, Seoul National University, Seoul 08826, South Korea}
\affil[2]{Institute of New Media and Communications, Seoul National University, Seoul 08826, South Korea}
\affil[3]{Automation and Systems Research Institute, Seoul National University, Seoul 08826, South Korea}
\affil[4]{Kakao Mobility, Gyeonggi-do 13494, South Korea}

\date{}

\begin{document}
\maketitle

\pagestyle{myheadings}
\thispagestyle{plain}

\begin{abstract}
In this paper, we discuss a large-scale fleet management problem in a multi-objective setting. 
We aim to seek a receding horizon taxi dispatch solution that serves as many ride requests as possible while minimizing the cost of relocating vehicles. 
To obtain the desired  solution, we first convert the multi-objective  taxi dispatch problem into a network flow problem, which can be solved using the classical minimum cost maximum flow (MCMF) algorithm.
We show that a solution obtained using the MCMF algorithm is integer-valued;  thus, it does not require any additional rounding procedure that may introduce undesirable numerical errors. 
Furthermore, we prove the {\it time-greedy property} of the proposed solution, which justifies the use of receding horizon optimization. 
For computational efficiency, we propose a linear programming method 
to obtain an optimal solution in near real time. 
The results of our simulation studies using the real-world data for the metropolitan area of Seoul, South Korea indicate that the performance of the proposed predictive method is  almost as good as that of the oracle that foresees the future.
\end{abstract}


\section{Introduction}

\label{sec:introduction}
On-demand ride-hailing services, such as Uber, Lyft, Didi, Grab, and Kakao T, have handled  millions of ride-hailing orders per day in 2019.
The increasing demand for ride-hailing services introduces various new technologies  and   challenges to  service providers~\cite{8713568, 7995029, 8412487}.
The most challenging problem is to maintain a proper number of ride service suppliers (taxis or drivers) in each service area.
If there are not enough vehicles located nearby when a customer requests a ride, then the service provider is unable to serve the request. 
This leads to bad user experience and  potential revenue loss.
Therefore, many ride-hailing platforms focus on developing ways to align the spatial distribution of vehicles (or drivers) to that of ride requests by using,  for example, dynamic surge pricing \cite{chen2016dynamic} and pooling \cite{mohlmann2017hands}, to improve the order fulfillment rate of each platform.

Aligning the  spatial distribution of supply to that of demand is a long-standing research topic in spatial crowdsourcing \cite{zhao2016spatial}.
In the context of ride-hailing services, these problems are closely related to  fleet management (e.g., \cite{godfrey2002adaptive,SchmidVerena2012Stda,HongLianxi2012AiLa,GhannadpourSeyedFarid2014Amdv,powell1998dynamic,george2011fleet}), and   taxi  dispatching   (e.g., \cite{zhang2017taxi,Liu:2019:GRS:3308558.3313579,Huang:2019:OTD:3306127.3331998,Tang:2019:DVB:3292500.3330724,lee2004taxi,ke2018hexagon}).

In taxi dispatching, the major concern  is to consider the future demand for ride requests, which is difficult to predict reliably. 
The demand is often modeled as a random process with an estimated distribution, and a dispatch solution is obtained by modeling a mobility-on-demand system of interest as a queueing network~\cite{zhang2014queueing, zhang2016control}, and a Markov decision process~\cite{volkov2012markov}.
However, incorporating the data collected in real time into these offline solutions is a nontrivial task. 
This issue has been alleviated by using (deep) reinforcement learning (RL)\cite{jintaoke2010optimizing,thomyphan2019distributed}. 
RL is intended to capture the interactions between a large volume of vehicles in an adaptive manner. 
However, due to the curse of dimensionality, in practice,   it is used in conjunction with an approximation technique, which often degrades the performance of this approach in large-scale fleet management~\cite{minneli2018efficient,yujiechen2019can}.
RL methods also require a substantial  amount of data to learn an efficient dispatch policy by capturing
how to utilize various factors in a given transportation system~\cite{oda2018movi,ke2018hexagon,linkaixiang2018efficient,xihanli2019cooperative}.
Moreover, it is difficult to incorporate constraints into RL methods, particularly when the constraints change over time, although several methods have recently been proposed to address this issue~\cite{achiam2017constrained,dalal2018safe,tessler2018reward}.

On the other hand, a model predictive control (MPC) method can also be used to deal with the dispatching problems in the mobility-on-demand system. The idea of the MPC method is to solve an optimization problem at each time step for a fixed time horizon, obtain a sequence of controls, and use only the first control for the current time step. In the next time step, the same procedure is performed, but with the shifted time horizon. MPC methods have demonstrated several theoretical and empirical advantages in solving large-scale taxi dispatch problems.
In MPC, online observations and constraints can simply be considered in a receding-horizon optimization problem in each time step.
Several MPC-based dispatch methods have been proposed in previous literature. To name a few, the MPC-based method is applied to the problem of rebalancing a bike-sharing system that ensures
the feasibility of imposed constraints~\cite{calafiore2017flow}, but they considered a degree of alignment between demand and supply as a constraint, not as an objective function. A taxi service planning is designed using the MPC method so that monotonic improvement of iterative solutions is guaranteed~\cite{luo2018dynamic}, but they only minimize cruising distance of the vehicle. Other than them, the MPC methods have been applied to provide a worst-case performance guarantee \cite{miao2019data}, the asymptotic performance guarantee that the number of unserved ride requests converges to zero~\cite{zhang2016model}, or combined with long short-term memory (LSTM) neural networks which forecast a future demand~\cite{iglesias2018data}.
Another variant of the MPC approach, which is closely related to our method, was proposed in~\cite{miao2016taxi}. They considered a multi-objective framework similar to our method, which tries to make a ratio between supply and demand on vehicle homogeneous over a spatial region with a minimal idle driving distance. They introduced a stochastic taxi dispatching model and achieved consistency between driver and passenger distributions. On the other hand, we consider a deterministic taxi dispatching model based on the flow problem. Our method has additional advantageous theoretical properties that remove rounding errors and provide a performance guarantee regardless of prediction accuracy.

Finally, we mention that there have been several previous works about allocating the ride requests to the currently idle taxis. For example, the allocation of the ride request can be modeled by a combinatorial optimization problem and the acceptance of the allocated request can be modeled by the probability distribution [13]. Moreover, an independent and cooperative ride request allocating algorithm was also proposed [24]. In [38], an RL-based algorithm that can allocate the large-scale ride requests in real-time was proposed. However, departing from the above references, we focus on repositioning drivers and the allocation of the ride request is automatically done by allocating from the most expensive requests to idle drivers.

In this paper, we propose an efficient, yet very simple method for large-scale taxi dispatch by exploiting the aforementioned advantages of MPC in a multi-objective setting.
The proposed taxi dispatch solution aims to serve as many ride requests as possible while minimizing the total reposition cost. 
We model an urban transportation system of interest as a {\it flow network}, which is a useful approach introduced to manage mobility-on-demand platforms \cite{contardo2012balancing,alvarezvaldes2016optimizing,calafiore2017flow}.
In particular, 
we model the transportation system  as a grid-world consisting of homogeneous regular hexagonal cells as  in \cite{ke2018hexagon,linkaixiang2018efficient,zhouming2019multi}.
Each hexagonal cell represents the coverage of a vehicle in a unit time interval.

To develop a tractable solution method, we convert the multi-objective taxi dispatch problem in a receding time window into a minimum cost maximum flow (MCMF) problem, which can be solved using well-known algorithms \cite{FordL.R.LesterRandolph2010Fin,BusackerRobertG1960APFD}.
We also investigate the theoretical properties of the proposed method. 
We first show that the optimal solution obtained by the proposed method is a vector of integers. 
Thus, our method provides a physically meaningful solution because the decision variable represents the number of relocating vehicles, and it does not require any rounding, which may introduce undesirable numerical errors. 
Moreover, we show the {\it time-greedy property} that justifies the use of receding horizon optimization. The time-greedy property of the solution means that the proposed dispatch solution maximizes the number of served ride requests at least at the initial time step, even under uncertainty on future demand. In other words, the number of served ride requests at the initial time step using the proposed dispatch solution is always greater than equal to that of any other dispatch solution. The time-greedy property provides a worst-case performance guarantee and therefore, the proposed predictive method limits the undesirable negative impact of future uncertainty in demand.
To further improve computational efficiency for large-scale taxi dispatch problems, we also develop a {\it linear programming} (LP) method
by reformulating  the MCMF problem as a linear optimization problem without loss of optimality.
According to the results of our numerical test,
the proposed LP method is approximately 880 times faster than the classical MCMF algorithm~\cite{FordL.R.LesterRandolph2010Fin,BusackerRobertG1960APFD}.

The performance  of the proposed taxi dispatch method is demonstrated using real-world data for the metropolitan area of Seoul, South Korea, and it is compared with the performance of the oracle, a deep RL method, and two rule-based algorithms.
The performance of the proposed method is comparable to that of the oracle although the oracle foresees the future. Specifically, it achieves 97.16\% of the oracle's performance.
Consequently, our dispatch solution significantly outperforms all the other methods. 
 
 We summarize the main contribution of our paper as follows.

\begin{itemize}
	\item We present the flow network problem equivalent to the multi-objective taxi dispatch problem. The solution of the flow network problem is automatically guaranteed as an integer solution. 
	\item We propose a model predictive control method to consider potential future ride requests. The MPC method satisfies time-greedy property which guarantees that the method takes a maximal number of requests, at least at the current time step regardless of the accuracy of prediction.
	\item We further convert a flow network problem into the equivalent LP formulation, which has significantly less computational cost than the flow network problem.
	\item Finally, we make a simulation and compare the performance of our proposed algorithm with the other algorithms, using the real-world data of Seoul. The proposed algorithm shows almost the same performance as the oracle and outperforms a deep RL-based method and two rule-based algorithms that are used to be compared.
\end{itemize}

The rest of the paper is organized as follows. In Section \ref{sec:setup}, the taxi dispatch problem is formulated as a multi-objective optimization problem. In Section \ref{sec:method}, we convert this dispatch problem into the MCMF problem, and introduce the classical MCMF algorithm. 
We also provide the theoretical properties of the optimal solution obtained using the MCMF algorithm. 
For computational efficiency, we  also propose an equivalent LP formulation in this section. 
In Section \ref{sec:Experiment results}, we present  the results of our numerical experiments using real-world data, and discuss comparisons with other methods.

\section{The Setup}\label{sec:setup}

\begin{figure}[tb]
	\centering
	\includegraphics[width=3in]{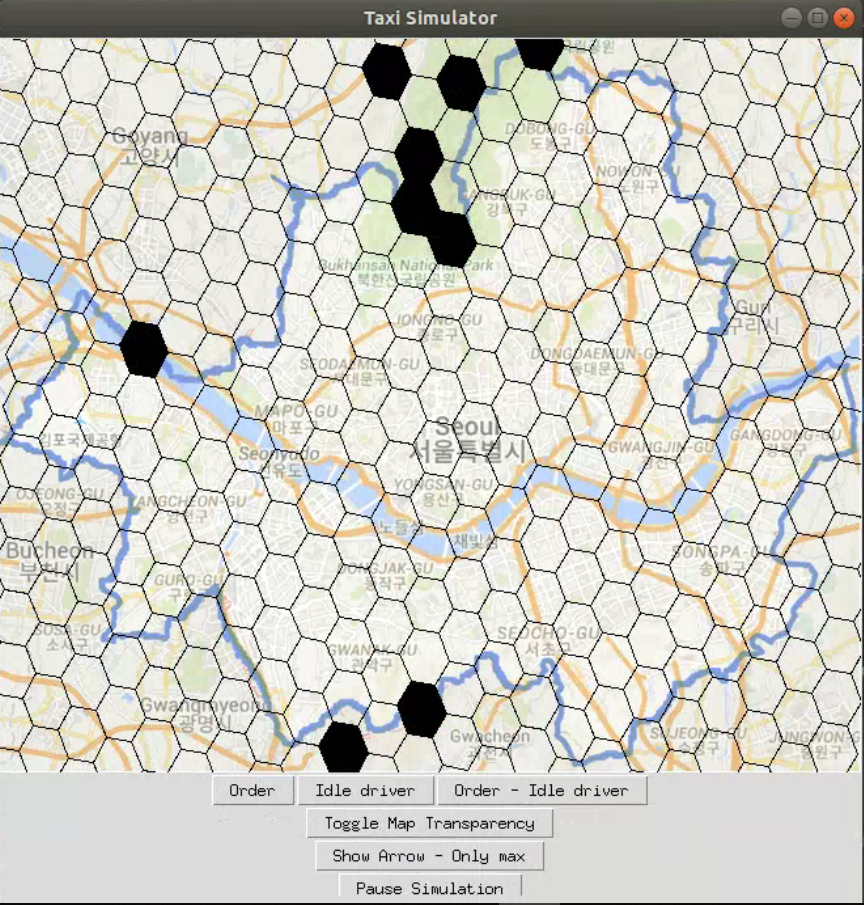}
	\caption{Hexagonal grid map ofthe  Seoul metropolitan area.
	The boundary of Seoul City is shown in light blue on the map. 
	}
	\label{hexagonal}
\end{figure}

In this paper, we consider a global taxi controller that assigns ride requests to homogeneous taxis  and relocates vehicles to another area,  if necessary. The goal of this controller is to maximize the gross merchandise volume (GMV), the sum of the fares of the served requests, while minimizing the repositioning costs.
A controller can choose to reposition a vehicle from a low-demand area to a high-demand area to maximize GMV, but this repositioning implies that a vehicle would use gas without a passenger, which is a cost.
Therefore, we model the objective of the global controller as maximizing the difference between GMV and the total repositioning costs.

At each time step, this fleet controller can command three actions to each taxi:
\begin{itemize}
\item To move a taxi to a nearby area (repositioning);
\item To assign a nearby  ride request (order assignment);
\item To wait (idle).
\end{itemize}

We note that, if there are not enough ride requests nearby, taxis can travel a long distance to pick up passengers by repeating repositioning action for several time steps.

Here, we assume that all taxis are homogeneous and always follow the command without exception.
When the controller assigns a ride request, which consists of origin, destination, fare, and estimated time of arrival to the destination, to a taxi, the vehicle goes offline at the origin, and it becomes available to follow the command after arriving at the destination (goes online).
Finally, when the controller decides to make a taxi idle, the vehicle is available to follow the control actions in the next time step at the vehicle's current location.
Throughout the paper, we interchangeably use the terms, taxi and vehicle.

\subsection{Fleet Management System Model}

We partition an area of interest into regular hexagons of the same size, and  we also uniformly partition a day into several intervals (usually by 10 minutes). Although one can consider a more accurate model, we will focus on this simple hexagonal model for several reasons. First, the size of the hexagonal grid is chosen so that the median time to travel from the center of one cell to its boundary is approximately 10 minutes, so that vehicles can move to the neighboring cell in a single time interval. Moreover, as we will see in the numerical experiment, the efficiency of the random method is reasonably close to our method. These imply that our model well approximates reality. Moreover, accurate simulation requires a massive amount of computing resources, which is infeasible. For this reason, previous researches also rely on the same idea of hexagonal discretization [18], [38].

Figure~\ref{hexagonal} shows a hexagonal grid map of
 the Seoul metropolitan area. 
We use this spatial discretization to model the dynamics of vehicle and order distributions.
Aside from the map boundary, each cell in the aforementioned hexagonal grid map has  6 adjacent neighboring cells. 
To model this connectivity in the hexagonal grid, we construct a graph, where a vertex is a cell in the hexagonal grid, denoted by $\mathcal{N}_i$, $i=1, \ldots, n$, and an edge exists between two vertices if the corresponding two cells are adjacent to each other in the hexagonal grid map.
More precisely, we define $L$ as an $n\times n$ {\it adjacency matrix} of nodes, given by
\[L_{ij}=\begin{cases}
1 & \quad \mbox{if $\mathcal{N}_i$ is adjacent to $\mathcal{N}_j$ or $i=j$}\\
0 & \quad \mbox{otherwise}.
\end{cases}\] 
Note that we allow the self-loop so that $L_{ii}=1$ for all $i$, which implies that the movement from one hexagonal region to itself is also considered. 
Thus, if $\mathcal{N}_i$ is located in the interior of the map, the $i$th row and the $i$th column of $L$ have exactly 7 ones, and the remaining rows and columns have less than 7 ones. For example, consider the 7 hexagonal cells shown in Figure \ref{grid}. The adjacency matrix $L$ for this grid is given by
\[L=\begin{pmatrix}
1&1&0&0&0&1&1\\
1&1&1&0&0&0&1\\
0&1&1&1&0&0&1\\
0&0&1&1&1&0&1\\
0&0&0&1&1&1&1\\
1&0&0&0&1&1&1\\
1&1&1&1&1&1&1
\end{pmatrix}\]
because $\mathcal{N}_7$ is adjacent to all the other cells. 

\begin{figure}[tb]
	\centering
	\includegraphics[width=1.5in]{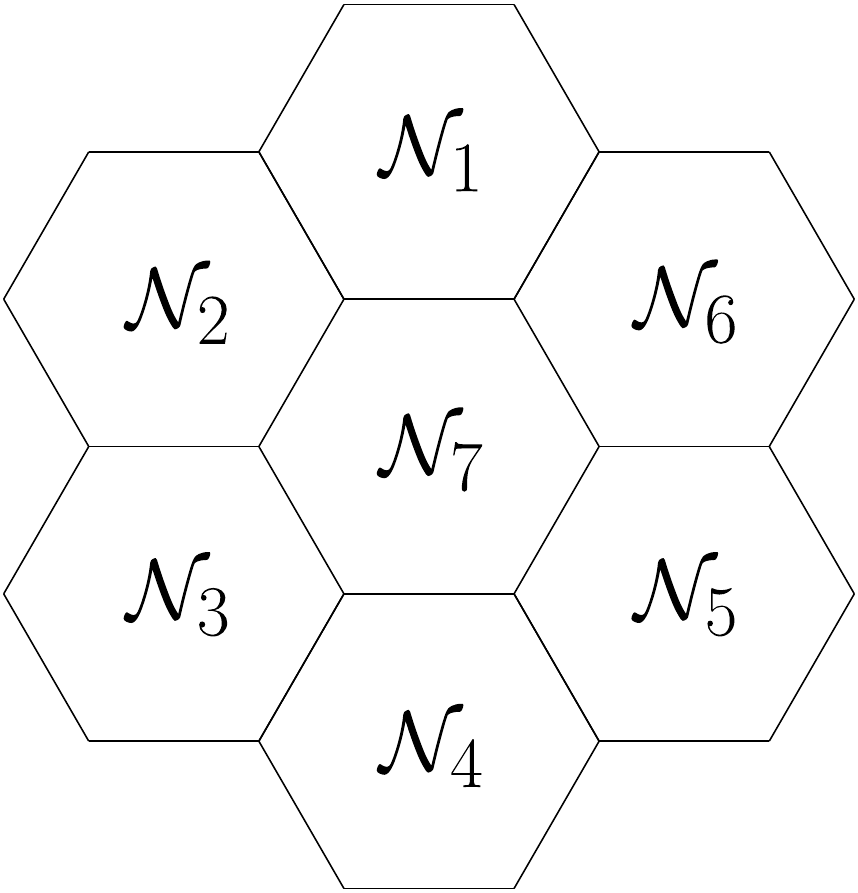}
	\caption{Simple hexagonal grid.}
	\label{grid}
\end{figure}

The time interval for two consecutive dispatch decisions is chosen as 10 minutes. 
Thus, one day is discretized into $T=144$ time steps.
We use the time index $t \in \{0,1,...,T-1\}$.
It is assumed that, in each time step, a vehicle can stay in the current cell or move to its neighboring cell.

\begin{figure}[tb]
	\centering
	\includegraphics[width=3.35in]{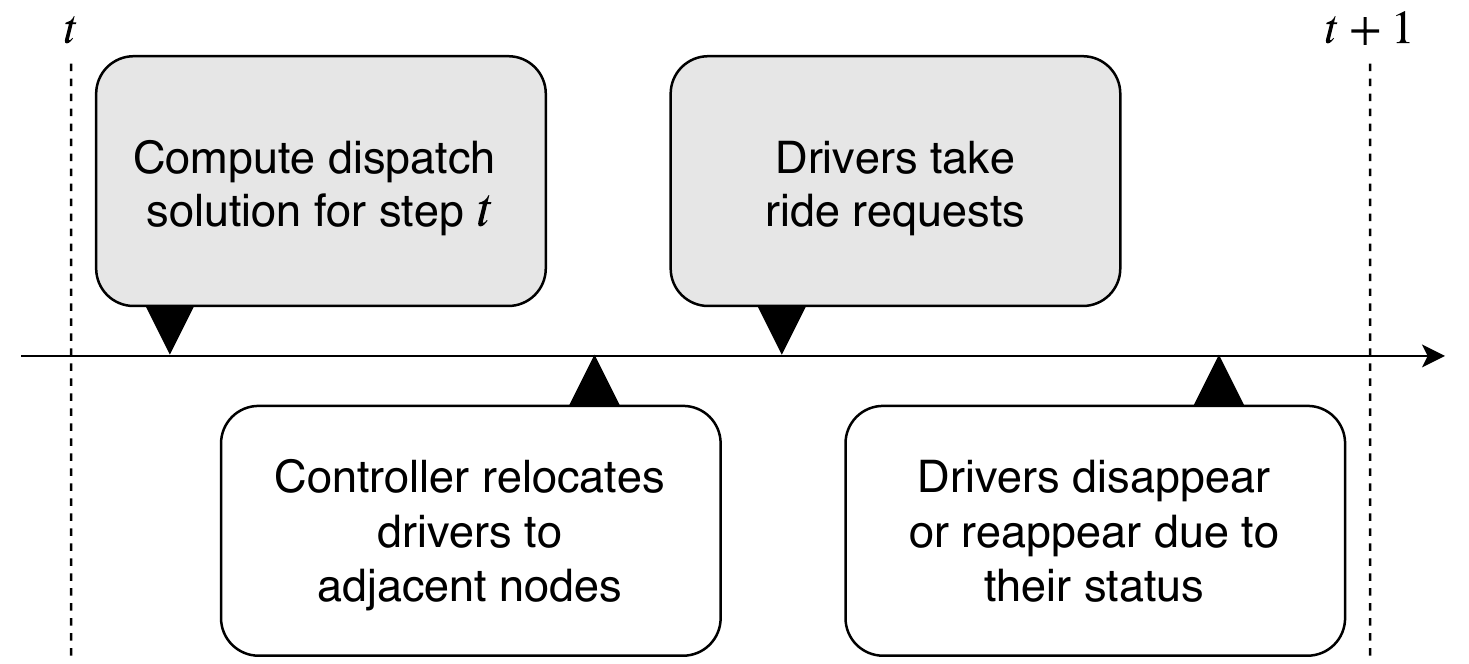}
	\caption{Illustration of fleet management processes in a single time step}
	\label{fig:dispatch}
\end{figure}

In the beginning of time step $t$,
the fleet controller has information about 
the number of ride requests $r_{i,t}$ and the driver distribution $d_{i,t}$. 
Given this information, it computes a taxi dispatch solution $x_t$.
Then, the vehicle redistribution process begins based on the dispatch solution, as illustrated in Figure~\ref{fig:dispatch}.
The process of redistributing vehicles
 in each time step consists of 
 three stages, as presented below:\footnote{
The division of the vehicle redistribution process into three stages  is somewhat artificial, but it is useful for formulating the system as a network flow problem, as discussed in Section \ref{sec:method}.
}
\begin{enumerate}
\item In the first stage, a vehicle located in $\mathcal{N}_i$ is capable of moving to $\mathcal{N}_j$ if $L_{ij}=1$. Thus, every vehicle can stay in the same hexagonal cell or move to the neighboring cells. Let $d_{i,t}$ denote the number of drivers or vehicles in $\mathcal{N}_i$ in the beginning of time step $t$. 
Then, the controller decides the number of vehicles that move from $\mathcal{N}_i$ to $\mathcal{N}_j$ in time step $t$, which is denoted as $x_{i\to j,t}$. 
If $L_{ij}=0$, then clearly $x_{i \to j,t} = 0$. 
Finally, the reposition cost $c_{i \to j,t}$ is imposed when a vehicle moves from $\mathcal{N}_i$ to $\mathcal{N}_j$ in time step $t$.

\item In the second stage, drivers take ride requests in the current cell. 
Let $r_{i,t}$
 denote the number of ride requests in $\mathcal{N}_i$ in time step $t$. Whenever a driver in $\mathcal{N}_i$ serves a request, it is excluded from $d_{i,t}$.
 After the travel is completed, the driver is added to the destination $\mathcal{N}_j$.

\item In the third stage, the total number of vehicles in $\mathcal{N}_i$ is adjusted based on each vehicle's status. 
Some vehicles may go offline due to the end of service, while some other vehicles may go online to start the service.
Moreover, any vehicle with a ride request just completed becomes online in the destination cell. 
\end{enumerate}
In the second stage, after vehicles are repositioned,
the controller automatically matches drivers and ride requests from the most expensive ones to maximize the income from the served requests. The controller neglects other factors, such as the expected travel time. 
We will justify this matching strategy in Section~\ref{sec:Experiment results} by empirically showing that the performance of our solution is similar to that of the oracle, which uses all data about ride requests.
We note that this process achieves fairness among drivers in the same grid by randomly assigning ride requests to them.
Under this {automated matching} assumption, maximizing GMV is equivalent to serving  as many  ride requests as possible. 
We consider the latter as an objective function in the optimization problem to be introduced.

The notation used throughout  the paper is summarized in Table \ref{table}.
We also use the following notation for vectors and matrices: 
$d_{t} := (d_{1,t}, \ldots, d_{n,t}) \in \mathbb{R}^n$, and $d := \begin{bmatrix} 
d_0 & \cdots & d_{T-1}
\end{bmatrix} \in \mathbb{R}^{n \times T}$.
Similar notation is used for vectors and matrices of $r_{i,t}$. 
Moreover, we let $x_t := (x_{1 \to 1, t}, \ldots, x_{1 \to n, t},$ $\ldots, x_{n \to 1, t}, \ldots, x_{n \to n, t}) \in \mathbb{R}^{n^2}$, and 
$x := \begin{bmatrix} 
x_0 & \cdots & x_{T-1}
\end{bmatrix} \in \mathbb{R}^{n^2 \times T}$.
Similar notation is used for vectors and matrices of $c_{i\to j,t}$.

\begin{table}[tb]
\normalsize
\centering
\begin{tabular}{ l|l }
\hline
Symbol& 
Quantity
 \\
\hline
\hline
$\mathcal{N}_i$& 
$i$th hexagonal cell (node) \\
$L$& 
adjacency matrix \\
$d_{i,t}$& 
\#  of drivers in $\mathcal{N}_i$ in step $t$ \\
$r_{i,t}$& 
\#  of ride requests in $\mathcal{N}_i$ in step $t$ \\
$x_{i \to j,t}$& 
\# of vehicles moving from $\mathcal{N}_i$ to $\mathcal{N}_j$ in step $t$ \\
$c_{i \to j,t}$& 
cost of relocating from $\mathcal{N}_i$ to $\mathcal{N}_j$ in step $t$\\

\hline
\end{tabular}
\caption{Notation}
\label{table}
\end{table}

\subsection{Multi-Objective Taxi Dispatch Problem}
\label{sec:2B}

The objective of taxi dispatch is to maximize GMV or, equivalently, to serve  as many  ride requests as possible, with minimal relocation costs. 
Thus, the goal of the fleet controller is twofold: $(i)$ to maximize the number of served requests, and $(ii)$ to minimize the total  cost incurred by relocating vehicles. 
Therefore, it is natural to consider a multi-objective optimization problem.
We first note that the number of served requests in $\mathcal{N}_i$ cannot be larger than the number of drivers or the number of ride requests in the same cell. 
The number of drivers in $\mathcal{N}_i$ in time step $t$ is given by
$d_{i,t}+\sum_{j=1}^n x_{j\to i,t}-\sum_{k=1}^n x_{i\to k,t}$.
Thus, the total number of served requests  in $\mathcal{N}_i$  in time step $t$ is given by
\[
g_{i, t} (x_t) := \min \bigg \{ d_{i,t}+\sum_{j=1}^n x_{j\to i,t}-\sum_{k=1}^n x_{i\to k,t},r_{i,t}\bigg \}.
\]
Hence, the 
first objective function, which represents the total number of served requests throughout the entier day, is defined by
\begin{equation} \label{opt1}
\bar{f}_+ (x):= \sum_{t=0}^{T-1} \sum_{i=1}^n   g_{i, t} (x_{t}).
\end{equation}

On the other hand, the cost of repositioning vehicles in time step $t$ is simply given by
\[
h_t(x_t) := \sum_{i,j=1}^n c_{i \to j, t} x_{i \to j, t}.
\]
Thus, the second objective function,  which computes the total repositioning cost, is defined by
\begin{equation}\label{opt2}
\bar{f}_- (x):=  \sum_{t=0}^{T-1} h_t (x_t).
\end{equation}
Note also that the following constraints must be satisfied:
\begin{subequations}\label{const}
\begin{align}
&d_{i,t+1} = d_{i, t}+\sum_{j=1}^n  x_{j \to  i,t}- \sum_{k=1}^n  x_{i\to k,t} - g_{i,t} (x_t) \label{const:flow} \\
&\sum_{j=1}^n x_{i \to j,t} \leq d_{i,t} \label{const:mass}\\
& x_{i \to j,t}\geq 0 \quad \forall i=1, \ldots, n, \forall t = 0, \ldots, T-1 \label{const:sign}
\end{align}
\end{subequations}
Here, the first constraint describes the evolution of the number of vehicles in $\mathcal{N}_i$.
The second constraint represents that the vehicles moving out of $\mathcal{N}_i$ is limited by the number of drivers in the cell.

Putting all the pieces together, we formulate the taxi dispatch problem as the following multi-objective optimization problem:
\begin{equation} \label{opt}
\begin{split}
\max \quad & \left (\bar{f}_+ (x), - \bar{f}_- (x) \right )\\
\mbox{s.t.} \quad &\eqref{const}.
\end{split}
\end{equation}
The major concern in solving this problem is the lack of knowledge about the number of ride requests $r_{i,t}$ for future time steps. 
As time goes by, information about $r_{i, t}$'s unfolds. 
Thus, instead of directly solving this problem offline,
we optimize the objective functions in a {\it receding horizon} fashion.
  In the next section, we explain the receding horizon optimization problem in detail, and we solve it by using an equivalent minimum cost maximum flow problem and a linear program.

\section{Predictive Dispatch Method via Network Flow Optimization}\label{sec:method}

\subsection{Converting to A Network Flow Problem}

We first consider the simplest case when $T = 1$. 
In this case, the constraint~\eqref{const:flow} has no effect. 
Recall that the fleet controller redistributes the vehicles in each hexagonal cell to its adjacent cells following the three-stage procedure. 
We interpret this procedure as controlling the flow between the nodes.
Let $\mathcal{V}_i$ and $\mathcal{W}_i$ be the copies of nodes $\mathcal{N}_i$, $i=1, \ldots, n$.
We consider
the movement of vehicles  at the first stage as the flow from the nodes $\mathcal{V}_i$ to $\mathcal{W}_j$.
Here, the flow capacity between $\mathcal{V}_i$ and $\mathcal{W}_j$ is defined as $+\infty$ if $L_{ij}=1$ and 0 otherwise. 
Moreover, since no vehicles can move from the future to the past, the flow capacity from $\mathcal{W}_j$ to $\mathcal{V}_i$ is set to be 0, regardless of the value of $L_{ij}$. 

\begin{figure}[tb]
	\centering
	\includegraphics[width=2.5in]{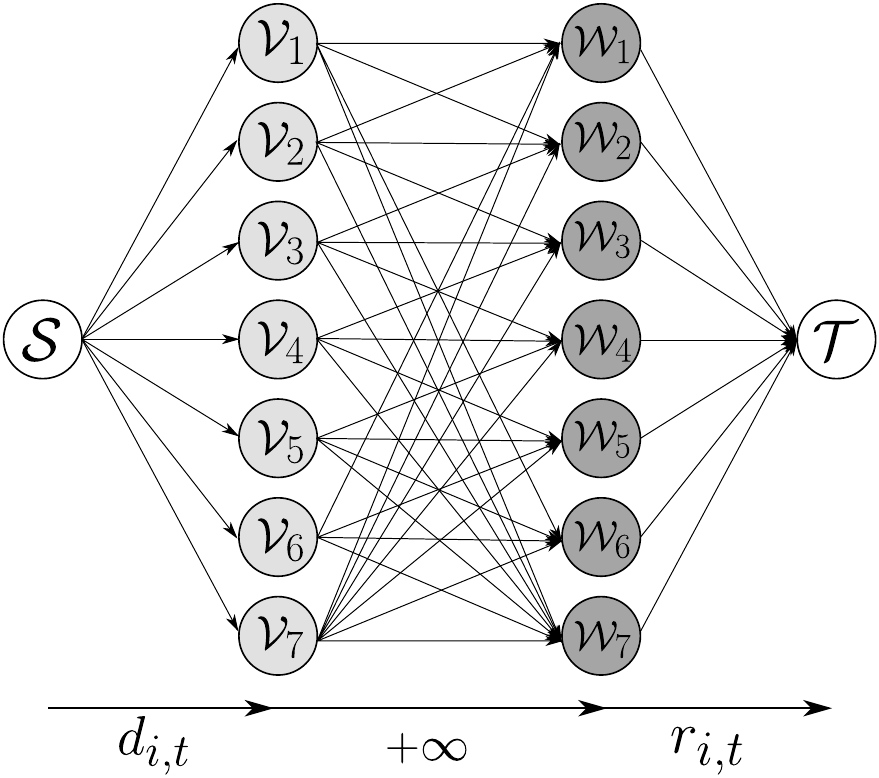}
	\caption{Flow diagram of the taxi dispatch problem with 7 hexagonal cells in Figure~\ref{grid} when $T=1$.}
	\label{flow}
\end{figure}

To complete the network flow problem formulation, we introduce an artificial source, $\mathcal{S},$ and a sink, $\mathcal{T}$. 
The flow capacity from  $\mathcal{S}$ to $\mathcal{V}_i$ corresponds to the number of vehicles in the beginning of the first stage, which is equal to $d_{i, 0}$, and the flow capacity from $\mathcal{W}_j$ to $\mathcal{T}$ corresponds to the number of requests, which is $r_{i,0}$. Finally, the cost of moving from $\mathcal{V}_i$ to $\mathcal{W}_j$ is $c_{i\to j, 0}$, and the cost of moving from  $\mathcal{S}$ to $\mathcal{V}_i$ or from $\mathcal{W}_j$ to $\mathcal{T}$ is 0. Figure \ref{flow} shows 
 the flow diagram converted from the original problem with  the 7 hexagonal regions illustrated in Figure \ref{grid}. 
The one-to-one relationships between the taxi dispatch problem and the network flow problem are summarized in Table~\ref{table2}.

\begin{table}[tb]
	\normalsize
	\centering
	\begin{tabular}{l || l|l}
		\hline
		&Taxi dispatch & Network flow problem 
		\\
		\hline
				\hline
		$d_{i,t}$& \# of drivers  & flow capacity from $\mathcal{S}$ to $\mathcal{V}_i$\\
		$r_{i,t}$&  \# of  requests & flow capacity from $\mathcal{W}_i$ to $\mathcal{T}$\\
		$c_{i \to j, t}$ & reposition costs & flow costs\\
		\hline
	\end{tabular}
	\caption{Relationship between the taxi dispatch problem and the network flow problem}
	\label{table2}
\end{table}

In this network flow setup, we can use the minimum cost maximum flow (MCMF) algorithm introduced in \cite{FordL.R.LesterRandolph2010Fin,BusackerRobertG1960APFD}.
To  describe the MCMF algorithm, we introduce several concepts about  directed flow networks. 
Let $G=(V,E)$ be a digraph with a  set of vertices $V=\left\{v_1,v_2,\ldots, v_n\right\}$ that includes a source, $\mathcal{S}$, and a sink, $\mathcal{T}$,  and the  set of edges $E\subseteq V\times V$. Assume that the capacity $s_{ij}$ and cost $c_{ij}$ associated with each edge $(v_i,v_j)$ are given. 
Let $f_{ij}$ denote the flow from $v_i$ to $v_j$.
We now define the {\it residual network} and {\it augmented path} as follows:
\begin{definition}
Suppose that a flow network $G = (V, E)$ is given.
	\begin{enumerate}
		\item 
		The \emph{residual network} $G_f = (V, E_f)$ is defined by the graph with the node set $V$, and edge set $E_f := \{ (v_i,v_j) \in V \times V \mid s^f_{ij} :=s_{ij}-f_{ij} > 0\}$.
		
		\item In the residual network $G_f$, a path from the source to sink $p := ( \mathcal{S}=v_{i_1},v_{i_2},\ldots, v_{i_k}= \mathcal{T})$ is called an \emph{augmented path} if $s^f_{i_\ell  i_{\ell+1}}>0$ for all $\ell=1,2,\ldots,k-1$.
	\end{enumerate} 
\end{definition}
Here, $s_{ij}^f$ represents the capacity of edge $(v_i,v_j)$ in the residual network.

\begin{algorithm}[tb]
	\caption{Minimum cost maximum flow (MCMF) algorithm}
	\label{max_flow_algorithm}
	\begin{algorithmic}[1]
		\REQUIRE digraph $G=(V,E)$;
		\REQUIRE capacity $s_{ij}$, cost $c_{ij}$ for $(v_i,v_j)\in E$;
		\REQUIRE $(v_j,v_i)\notin E$  $\forall (i,j)$ such that  $(v_i,v_j)\in E$;
		\STATE Initialize $s_{ji}\gets0$, $c_{ji}\gets-c_{ij}$, $f_{ij}\gets0$, $f_{ji}\gets0$ for $(i,j)$ such that $(v_i,v_j)\in E$;
		
		\WHILE{ $\exists$ augmented path}
		\STATE $p \in \arg\min_{q:\textup{augmented path}} \sum_{(v_i,v_j)\in q} c_{ij}$;
		\STATE $s^f\gets\min\big \{s^f_{ij}\mid (v_i,v_j)\in p\big \}$;
		\FOR{$(i,j)$ such that $(v_i,v_j)\in p$} 
		\STATE $f_{ij}\gets f_{ij}+s^f$;
		\STATE $f_{ji}\gets f_{ji}-s^f$;
		\ENDFOR
		\ENDWHILE
	\end{algorithmic}
\end{algorithm}

The MCMF algorithm (Algorithm~\ref{max_flow_algorithm}) computes the maximum flow, while minimizing the total cost.
It is well known that the MCMF algorithm has  time complexity $O(|V|\times |E|\times F)$, where $F$ denotes the maximum amount of flow from $\mathcal{S}$ to $\mathcal{T}$~\cite{BusackerRobertG1960APFD}. 

By the equivalence between the taxi dispatch problem~\eqref{opt} and the network flow problem, the former can be solved  using Algorithm \ref{max_flow_algorithm}. 
For example, consider the network shown in Figure \ref{flow}. Here, the vertex set is $V=\{\mathcal{S},\mathcal{T},\mathcal{V}_i,\mathcal{W}_i \mid i=1,\ldots,7 \}$, and the edge set is given by $E=\{(\mathcal{S},\mathcal{V}_i),(\mathcal{W}_j,\mathcal{T}), (\mathcal{V}_{i'},\mathcal{W}_{j'}) \mid i,j=1,\ldots, 7, \; i', j' = 1, \ldots, 7 \mbox{ such that }  L_{i'j'} = 1 \}$. The capacities of $(\mathcal{S},\mathcal{V}_i)$ and $(\mathcal{W}_j,\mathcal{T})$ are set to be $d_{i, 0}$ and $r_{i, 0}$, respectively, and the capacity of $(\mathcal{V}_i,\mathcal{W}_j)\in E$ is set to be $+\infty$. Finally, the cost of sending flow is set to be $c_{i \to j, 0}$ for $(\mathcal{V}_i,\mathcal{W}_j)\in E$ and zero for all the other edges. 
Then, the following proposition is a direct result of the MCMF algorithm, applied to the taxi dispatch problem.

\begin{proposition}\cite[Theorem 4]{BusackerRobertG1960APFD}\label{P1}
	Let $x^\star := (x_{i\to j,0}^\star)_{i,j}$ be the minimum cost maximum flow solution obtained by applying the MCMF algorithm (Algorithm~\ref{max_flow_algorithm})  to the network flow problem converted from the taxi dispatch problem~\eqref{opt}. Then, $x^\star$ is an optimal solution of the taxi dispatch problem \eqref{opt}.
\end{proposition}

\subsection{Receding Horizon Optimization}\label{sec:3B}

In this subsection, we extend the idea introduced in the previous subsection to the receding horizon optimization problems for taxi dispatch. 
To set up the receding horizon optimization problem (with prediction horizon $K$) associated with our multi-objective taxi dispatch problem, 
let 
\begin{equation} \label{obj1}
{f}_{+,t} (x):= \sum_{k=t}^{t + K -1} \sum_{i=1}^n   g_{i, k} (x_{k}), \quad t=0,\ldots, T-K,
\end{equation}
and
\begin{equation}\label{obj2}
\tilde{f}_{-, t}(x) := \sum_{k=t}^{t + K-1} h_k (x_k), \quad t=0,\ldots, T-K,
\end{equation}
which are modified from the original objective functions \eqref{opt1} and \eqref{opt2}, respectively.
The objective functions
$f_{+,t}$ and $\tilde{f}_{-,t}$ represent the number of served requests  and the repositioning cost, respectively, from time step $t$
 during the prediction horizon.
 With a slight abuse of notation, we let $x := (x_t, \ldots, x_{t+K-1})$ in the receding horizon setting. 
Since  the number of ride requests for future time steps is not available, 
in time step $t$ we let 
\begin{equation}\label{const:ride}
r_{k} = r_{t}  \quad \forall k=t+1, \ldots, t+K-1.
\end{equation}
This choice is reasonable when the distribution $r_t$ of ride requests  does not change much during the prediction horizon.

We further modify the second objective function $\tilde{f}_{-,t}$ by adding a regularizer to serve as many current requests as possible, rather than expecting uncertain future requests:
\begin{equation}\label{obj3}
{f}_{-, t}(x) := \sum_{k=t}^{t + K-1} h_k (x_k) + \alpha  \sum_{k = t}^{t + K -1} \sum_{i=1}^n (k - t) g_{i,k} (x_k)
\end{equation}
for $t = 0, \ldots, T-K$. 
Here, $\alpha (k - t)$ denotes an artificial cost for serving a request in time step $k$.\footnote{Throughout the paper, the value of $\alpha$ is set to be $100$.} 

Given the observations $d_t, r_{t}$ in the beginning of time step $t$,
the receding horizon optimization problem for taxi dispatch can then be formulated as follows:
\begin{equation}\label{rh_opt}
\begin{split}
\max \quad & \left ( f_{+, t} (x), - f_{-, t} (x) \right )\\
\mbox{s.t.} \quad & \eqref{const}
\end{split}
\end{equation}
where all the constraints must hold for the prediction horizon. 
After obtaining an optimal solution $x^\star := (x_t^\star, \ldots, x_{t+K-1}^\star)$, we use only $x_t^\star$ to dispatch the currently available vehicles, as in the model predictive control (MPC) methods. 
In fact, this problem can be considered to be an MPC problem with state $d_t$, control input $x_t$, and system dynamics~\eqref{const:flow} and constraints~\eqref{const:mass} and \eqref{const:sign}.

\begin{figure}[tb]
	\centering
	\includegraphics[width=3.9in]{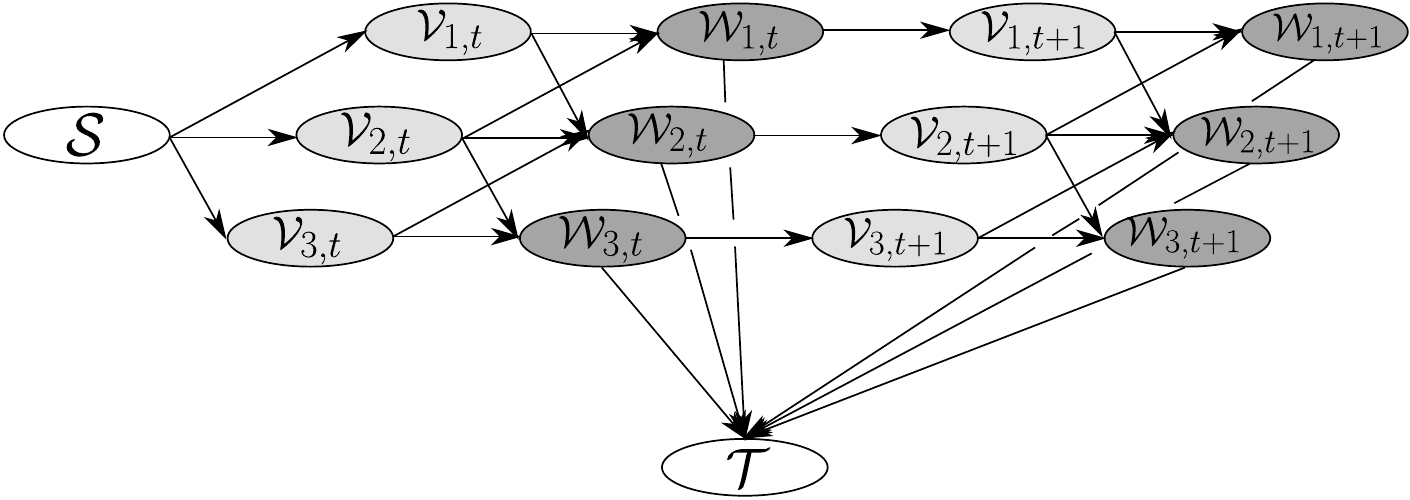}
	\caption{Flow diagram of the receding horizon optimization problem for $K=2$.}
	\label{flow_2}
\end{figure}
 
We now explain the equivalence between \eqref{rh_opt} and the network flow problem, which is similar to that seen in the stationary case.
 Let $\mathcal{V}_{i,k}, \mathcal{W}_{i, k}$ be the copies of $i$th node $\mathcal{N}_{i}$ for time step $k \in \{t,\ldots, t+K-1\}$. As before, the flow can move from $\mathcal{V}_{i, k}$ to $\mathcal{W}_{j, k}$ and from $\mathcal{W}_{i, k}$ to $\mathcal{V}_{i, k+1}$ without limitation, and it can move from $\mathcal{S}$  to $\mathcal{V}_{i,k}$ with a flow capacity $d_{i, k}$. However, since the ride requests can be matched in any of the  time steps, the flow can move from $\mathcal{W}_{i, k}$ to $\mathcal{T}$ with  capacity $r_{i, k}$. Lastly, we again attach the reposition cost $c_{i\to j, k}$ to the edge between $\mathcal{V}_{i,k}$ and $\mathcal{W}_{j,k}$ and the artificial cost $\alpha(k-t)$ to the edge between $\mathcal{W}_{i,k}$ and $\mathcal{T}$. For example, the extended flow diagram for $K=2$ is illustrated in Figure \ref{flow_2}.
Recall that a single time step 
 consists of three stages, as illustrated in Figure~\ref{fig:dispatch}.
We can connect the third stage of each time step to the first stage of the next time step to obtain the extended digraph.
The flow from $\mathcal{S}$ to $\mathcal{V}_{i,t}$ represents the current distribution of idle drivers (in the beginning of time step $t$).
In the subsequent time steps in the prediction horizon, when vehicles are unused in a time step, they become available in the next step.
This relationship is shown as the flow from $\mathcal{W}_{i, k}$ to $\mathcal{V}_{i, k+1}$.
Furthermore, as previously discussed, the flow from $\mathcal{V}_{i, k}$ to $\mathcal{W}_{j, k}$ represents the vehicles' motion in the first stage, and the flow from $\mathcal{W}_{i,k}$ to $\mathcal{T}$ presents taking the ride requests. Thus,
the receding horizon taxi dispatch problem \eqref{rh_opt} is equivalent to an MCMF  problem considered in the extended digraph.  
Again, we can obtain an optimal dispatch solution by applying Algorithm 1 to the network flow problem with the extended graph.
In the following two subsections, we examine two significant properties of the optimal dispatch solution obtained using the MCMF algorithm.

\subsection{Integer Property}

The first property of the optimal dispatch solution $x^\star$ is that all of its entries are integers. 
This integer property is an important advantage of the proposed taxi dispatch method.
In practice, it is physically impossible to use a non-integer dispatch solution. 
Thus, when using other methods that cannot guarantee the integer property,   an additional rounding procedure is required to convert the solution into a vector of integers. 
However, this procedure will introduce an undesirable error and lead to a sub-optimal dispatch solution.

The following theorem holds for a general network flow problem with integer capacity. 
\begin{theorem}\label{T1}
	Suppose the capacities in the graph $G=(V,E)$ are  integers. Then, the optimal flow plan constructed by the MCMF algorithm is a vector of  integers.
\end{theorem}

\begin{proof}
	We use an induction on the iteration of augmented paths. Let $p_1,\ldots, p_N$ be minimum cost augmented paths that are selected in each iteration of Algorithm \ref{max_flow_algorithm}. We will show that the flow vector remains as an integer vector after the $n$th iteration is finished for all $n$. In Algorithm \ref{max_flow_algorithm}, we initialize the flow vector as 0. During the first iteration, we update the flow $f_{ij}$ and $f_{ji}$ on the edges of $p_1$ by $s^f=\min\{s_{ij}^f\mid (v_i,v_j)\in p_1\}$. Since the capacities $s_{ij}$ of the edges and the initial flow quantities $f_{ij}\equiv 0$ are assumed to be integers, the residual capacities $s_{ij}^f$ of $p_1$ are also integers. Therefore, the updated entries $f_{ij}+s^f$ and $f_{ji}-s^f$ of the flow vector are also integers after the first iteration. Now, suppose the flow vector $f$ is an integer vector after  the $n$th iteration is finished. During the $(n+1)$th iteration, we update the flow $f_{ij}$ and $f_{ji}$ on the edges of $p_{n+1}$ by $s^f=\min\{s_{ij}^f\mid (v_i,v_j)\in p_{n+1}\}$. By using the same argument as $p_1$, the residual capacities $s_{ij}^f$ of $p_{n+1}$ are also integers. Therefore, the updated entries $f_{ij}+s^f$ and $f_{ji}-s^f$ of the flow vector are again integers.
	Since we repeat the procedure until the algorithm terminates, the resulting optimal flow plan should also be a vector of integers. 
\end{proof}
As a corollary, we have the integer property for the proposed dispatch solution $x^\star$.
\begin{corollary}
	Let $x^\star$ be the optimal dispatch plan obtained by applying the MCMF algorithm to the receding horizon optimization problem~\eqref{rh_opt}. Then, all the entries of  $x^\star$ are integers.  
\end{corollary}

\subsection{Time-Greedy Property}

We now highlight another important property of the optimal dispatch solution constructed using the MCMF algorithm. 
Recall that the motivation of introducing the  regularizer  in \eqref{obj3} is to serve as many requests as possible sooner rather than later in order to reduce the negative impacts of future uncertainty. 
In fact, we can show that the optimal dispatch solution from the MCMF algorithm of the receding horizon optimization problem maximizes the number of served requests in the current time step $t$ if the reposition costs are negligible.
To prove this {\it time-greedy} property, we need the following lemma.

\begin{lemma}\cite[Lemma 10]{BusackerRobertG1960APFD}\label{L1}
	The minimal cost augmented path  in Algorithm \ref{max_flow_algorithm} does not include a cycle.
\end{lemma}

By using Lemma~\ref{L1}, we can show the following time-greedy property of our optimal dispatch solution.

\begin{theorem}\label{T2}
	Suppose there is no reposition cost, i.e., $c =0$, and that $\alpha > 0$. Let $x^\star := (x_t^\star, \ldots, x_{t + K-1}^\star)$ be the optimal dispatch solution that is obtained by applying
 the MCMF algorithm to the receding horizon optimization problem~\eqref{rh_opt}.
 Then, $x_t^\star$ maximizes the number of served requests in time step $t$.
\end{theorem}

\begin{proof}
	Let $N$ be a total number of iterations in Algorithm \ref{max_flow_algorithm} and $p_1,\ldots, p_N$ be the selected augmented paths with minimum cost at each iteration. We first show that there is no edge of the form $(\mathcal{T},\mathcal{W}_{i,k})$ in every $p_n$. Suppose there exists  $(i,k)$ such that $(\mathcal{T},\mathcal{W}_{i,k})\in p_n$ for some $n$. On the other hand, by definition, every augmented path should end with $\mathcal{T}$. Hence, there exists another index pair $(j,l)$ such that $p_n$ ends with the edge $(\mathcal{W}_{j,l},\mathcal{T})$. Therefore, the path $p_n$ should have a cycle of the form $(\mathcal{T},\mathcal{W}_{i,k},\ldots,\mathcal{W}_{j,l},\mathcal{T})$. This contradicts to Lemma \ref{L1};  therefore any minimum cost augmented path $p_n$ should not have an edge of the form $(\mathcal{T},\mathcal{W}_{i,k})$. 
			
			Next, we show that the cost of $p_n$ is nonnegative. Suppose that the cost of $p_n$ is negative for some $n$. Since $c=0$, the cost is imposed only on the edges of the form $(\mathcal{W}_{i,k},\mathcal{T})$ where $t<k\le t+K-1$. Therefore, in order for $p_n$ to have a negative cost, it should have at least one edge of the form $(\mathcal{T},\mathcal{W}_{i,k})$. This contradicts to the fact that $p_n$ does not have an edge of the form $(\mathcal{T},\mathcal{W}_{i,k})$. Thus, the cost of $p_n$ is nonnegative for all $n =1, \ldots, N$.
			
			Since the costs of $p_n$ are nonnegative,   two cases are possible. The first case is that every $p_n$ has zero cost, and the second case is that there exists $1\le n_0\le N$ such that the costs of $p_1,\ldots, p_{n_0-1}$ are $0$ and the cost of $p_{n_0}$ is positive.
			
			For the first case, we will show, by induction, that every $p_n$ can only include $\mathcal{V}_{i,t}$ or $\mathcal{W}_{j,t}$, i.e., the nodes in time step $t$. 
			Let us first consider $p_1$. Since the flow was initialized as $0$ in Algorithm \ref{max_flow_algorithm}, and the capacities of $(\mathcal{W}_{j,k},\mathcal{V}_{i,k})$ are set to be 0, $p_1$ cannot contain an edge of the form $(\mathcal{W}_{j,k},\mathcal{V}_{i,k})$ because $s^f_{\mathcal{W}_{j,k}\mathcal{V}_{i,k}}=0$. Recall that the cost of $p_1$ was 0, which implies $(\mathcal{W}_{j,t},\mathcal{T})\in p_1$. Therefore, the only possible form of $p_1$ is $(\mathcal{S},\mathcal{V}_{i,t},\mathcal{W}_{j,t},\mathcal{T})$ for some $i,j$, and the first induction step holds for $p_1$. We now assume that $p_1,\ldots, p_n$ only include $\mathcal{V}_{i,t}$ or $\mathcal{W}_{j,t}$ and do not include $\mathcal{V}_{i,t+1}$. Since the edge $(\mathcal{W}_{j,t},\mathcal{V}_{i,t+1})$ is not contained in any of $p_1,\ldots, p_n$, the flow on it is still 0. Hence, $p_{n+1}$ cannot have any edge of the form $(\mathcal{V}_{i,t+1},\mathcal{W}_{j,t})$, which implies that once $p_{n+1}$ arrives at $\mathcal{V}_{i,t+1}$, it cannot move back to $\mathcal{W}_{j,t}$. However, since $p_{n+1}$ also has zero cost, it should end with the edge $(\mathcal{W}_{j,t},\mathcal{T})$. Therefore, $p_{n+1}$ can only include $\mathcal{V}_{i,t}$ or $\mathcal{W}_{j,t}$ and the induction holds for all $n =1, \ldots, N$. Until now, we have shown that every minimal cost augmented path $p_n$ contains only the vertices and edges of time step $t$. Therefore, in this case, the optimal dispatch solution is $x^\star=(x_t^\star,0,\ldots,0)$, and $x_t^\star$ is nothing but the optimal dispatch solution when $K=1$. Since the optimal dispatch solution when $K=1$ maximizes the number of served requests in time step $t$, the optimal dispatch solution $x_t^\star$ also maximizes the number of served requests in time step $t$.
			
For the second case. since $p_{n_0}$ is the first minimal cost augmented path with a positive cost, after the $(n_0-1)$th iteration is finished in Algorithm \ref{max_flow_algorithm}, there is no augmented path with zero cost. As in the first case, this implies that the dispatch solution after $(n_0-1)$th iteration should maximize the number of served requests in time step $t$. However, since the minimum cost augmented paths $p_n$ do not contain the edge of the form $(\mathcal{T},\mathcal{W}_{j,t})$, the number of served requests in time step $t$ cannot decrease during the iteration $n = n_0, \ldots, N$. Therefore, at the end of the iteration, the optimal dispatch solution $x^\star$ also maximizes the total number of served requests in time step $t$.
\end{proof}

\begin{remark}
	The time greedy property provides a worst-case performance guarantee in the following sense: the proposed solution $x_t^\star$ at least maximizes the number of served requests in time step $t$ even when the number of ride requests for future time steps is very different from the data used in the receding horizon optimization problem~\eqref{rh_opt}. 
On the other extreme in which the data well approximates the requests
 for future time steps in the prediction horizon, the proposed solution should perform even better based on the definition of the receding horizon optimization problem~\eqref{rh_opt}. 
 These features justify the effectiveness of the proposed MPC-based method.
\end{remark}

\subsection{Linear Programming Method}

To implement the proposed taxi dispatch method, 
the receding horizon optimization problem \eqref{rh_opt} must be solved within a fraction of the interval between two consecutive time steps. 
However, the running time for the MCMF algorithm is often long, particularly for the large-scale fleet management with a nontrivial prediction horizon.
To reduce the computation time to solve the taxi dispatch problem, 
we propose a linear programming (LP) approach.

We first note that the receding horizon optimization problem~\eqref{rh_opt} 
 is  a multi-objective optimization problem.
Since maximizing the number of served ride requests $f_{+, t}(x)$, and minimizing the total cost $f_{-,t}(x)$ are conflicting objectives, the set of optimal points of this problem forms a \textit{Pareto frontier}, at which there is no other solution that serves a greater number of requests with less cost. 
 Among the points on the Pareto frontier, we  focus on the point at which the number of served requests is maximized and  the cost is minimized (see Figure \ref{pareto}).
  This belongs to the so-called \textit{a priori} method for solving  multi-object optimization problems.

\begin{figure}[tb]
\centering
	\includegraphics[width=3in]{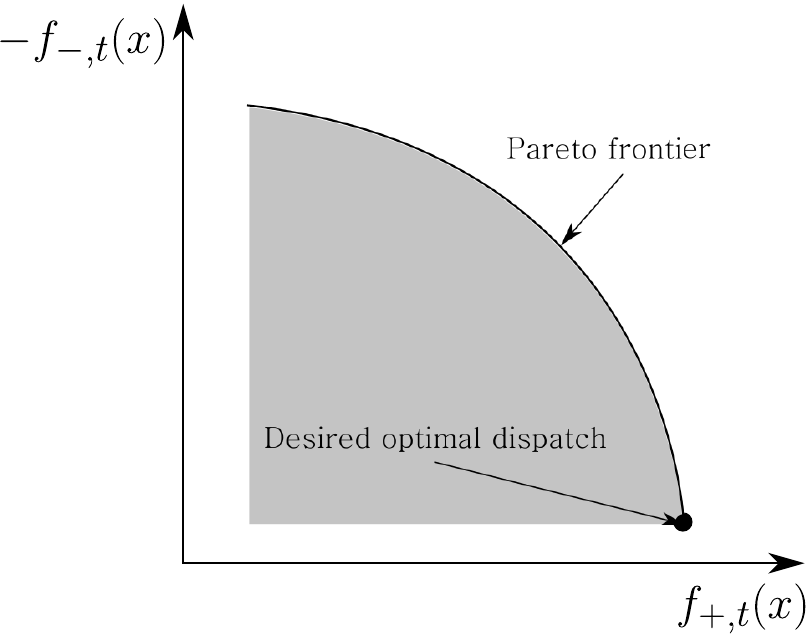}
	\caption{Pareto frontier and the desired optimal point.}
	\label{pareto}
\end{figure}

 The scalarizing technique is often used to find the desired optimal point with the \textit{a priori} selected preference~\cite{HwangMasud79}. 
Instead of considering the object functions $f_{+,t}$ and $-f_{-, t}$, we 
form the following optimization problem with a weighted sum of the two objectives ($\gamma > 0$):
\begin{equation} \label{scalar}
\begin{split}
\max_x \quad &
f_{+, t} (x) - \gamma f_{-, t} (x)\\
\mbox{s.t.} \quad & d_{i, k+1} = d_{i, k}+\sum_{j=1}^n  x_{j \to  i, k}- \sum_{l=1}^n  x_{i\to l, k} - g_{i, k} (x_k)\\
&\sum_{j=1}^n x_{i \to j, k} \leq d_{i, k}\\
& x_{i \to j, k}\geq 0, 
\end{split}
\end{equation}
where all the constraints are satisfied for $i,j=1, \ldots, n$ and $k = t, \ldots, t+K-1$. 
If  the positive scaling parameter $\gamma$ is sufficiently small,  an optimal solution of the reformulated problem~\eqref{scalar} is  the desired optimal dispatch plan.

\begin{theorem}\label{thm:lp}
		Suppose that 
		\[
		0 < \gamma < \frac{1}{\|f_{-,t} \|_{L^\infty}}.
		\]
		Let $\mathcal{X}$ be the feasible set of \eqref{scalar},
		 and let $x^\star$ be a maximizer of $(f_{+,t}(x)-\gamma f_{-,t}(x))$ in $\mathcal{X}$, i.e.,
		\[x^\star \in \argmax_{x\in \mathcal{X}}\left\{f_{+,t} (x)- \gamma f_{-,t} (x)\right\}.\]
		Then, 
		\[
		x^\star \in \argmax_{x\in\mathcal{X}}\left\{f_{+,t}(x)\right\}.
		\]
		 Moreover, if we let $\mathcal{X}_+ :=\arg\max_{x\in\mathcal{X}}\{f_{+,t}(x)\}$, then
		\[
		x^\star \in \argmin_{x\in \mathcal{X}_+ }\left\{f_{-,t}(x)\right\}.
		\]
		 Therefore, $x^\star$ is the desired dispatching solution that minimizes the reposition cost among the solutions maximizing the  number of served requests. 
	\end{theorem}
	
	\begin{proof}
Recall that $f_{+,t}(x)$ is the total number of served requests, which is an integer. We choose any $\tilde{x}\in \mathcal{X}_+$ and suppose that $x^\star \in \mathcal{X} \setminus \mathcal{X}_+$. We then have 
\begin{equation}\nonumber
\begin{split}
f_{+,t}(x^\star )-\gamma f_{-,t}(x^\star)&\le f_{+,t}(x^\star)\\
&\le f_{+,t}(\tilde{x})-1\\
&< f_{+,t}(\tilde{x})-\gamma f_{-,t}(\tilde{x}). 
\end{split}
\end{equation}
Therefore, $x^\star$ does not maximize $f_{+,t}(x)-\gamma f_{-,t}(x)$ over $\mathcal{X}$. This is a contradiction. 
Hence, we have $x^*\in \mathcal{X}_+$. 
Moreover, for any $\bar{x}\in \mathcal{X}_+$, we have 
\[
f_{+,t}(x^\star)-\gamma f_{-,t}(x^\star )\ge f_{+,t}(\bar{x})-\gamma f_{-,t}(\bar{x}).
\]
 Since both $x^*$ and $\bar{x}$ maximize $f_{+,t}(x)$ over $\mathcal{X}$, we have $f_{+,t}(x^\star)=f_{+,t}(\bar{x})$. Thus, we conclude that $f_{-,t}(x^\star)\le f_{-,t}(\bar{x})$. This implies $x^\star$ minimizes $f_{-,t}(x)$ over $\mathcal{X}_+$.
	\end{proof}

%Geometrically, this scalarization is equivalent to choosing a sufficiently stiff line whose optimal point is a desired optimal solution in Figure \ref{pareto}. 

By introducing a slack variable $z$, we can reformulate the problem \eqref{scalar} as a linear program.
\begin{proposition}\label{prop:lp}
Suppose that
\[
0 < \gamma< \frac{1}{\alpha(K-1)}.
\]
Then, the scalarized problem~\eqref{scalar} is equivalent to the following LP:
\begin{align}
\begin{aligned}\label{lp}
\max_{x,z} \; &\sum_{k=t}^{t+K-1} \sum_{i=1}^n z_{i, k} - \gamma \sum_{k=t}^{t + K-1} \bigg (\sum_{i,j=1}^n c_{i \to j, k} x_{i\to j,k}+\alpha \sum_{i=1}^n z_{i, k}(k-t)\bigg )\\
\mbox{s.t.} \quad &z_{i, k} \leq r_{i, k}\\
%&z_{i, k} \leq d_{i,k} +\sum_{j=1}^n  x_{j \to  i, k}- \sum_{l=1}^n  x_{i\to l, k} \\
&d_{i, k+1}\leq d_{i, k}+\sum_{j=1}^n  x_{j \to  i, k}- \sum_{l=1}^n  x_{i\to l, k} - z_{i, k} \\
&\sum_{j=1}^n x_{i \to j, k} \leq d_{i, k}\\
& x_{i \to j, k}\geq 0, 
\end{aligned}
\end{align}
where all the constraints are satisfied for $i,j=1, \ldots, n$ and $k = t, \ldots, t+K-1$, and $d_{i, t + K}$ is defined as $0$. 
\end{proposition}
It is the standard result of equivalent forms using slack variables in optimization; thus, we omitted the proof.

\begin{table}[tb]
	\centering
	\normalsize
	\begin{tabular}{ l | p{0.85in}   | p{0.85in} }
		\hline
		$K$&MCMF&LP
		\\
		\hline
				\hline
		1&0.232&0.033\\
		2&2.058&0.045\\
		4&17.292&0.068\\
		8&246.642&0.255\\
		16&997.343&1.129\\
		30&memory dump&2.491\\
		\hline
	\end{tabular}
%	}
	\caption{Computation time (in seconds) of the MCMF algorithm (Algorithm 1) and the LP method for solving \eqref{rh_opt} with different prediction horizon lengths.}
	\label{1023computation}
\end{table}

By  Theorem~\ref{thm:lp} and Proposition~\ref{prop:lp},
 the LP \eqref{lp} with $\gamma \in (0, \min \{\frac{1}{\|f_{-,t} \|_{L^\infty}}, \frac{1}{\alpha(K-1)}\} )$ is  equivalent to the original multi-objective receding horizon optimization problem.\footnote{We choose $\gamma = 10^{-5}$ and $\alpha = 100$ for simulations in the next section.}
The reformulated optimization problem~\eqref{lp} can be solved using several existing LP algorithms, such as simplex and interior point methods (see, e.g., \cite{Dantzig1998, Bertsimas1997, Vanderbei2015} and the references therein). 
In this present work, we used CPLEX among various LP solvers.\footnote{The default algorithm in CPLEX uses the interior point method which does not guarantee an integer solution. If the solution obtained by CPLEX includes a non-integer value, we can run a simplex algorithm to obtain an integer-valued optimal solution.}
It is difficult to efficiently accelerate 
the MCMF algorithm, whereas we can parallelize the LP when it is implemented with CPLEX \cite{manual1987ibm}.
As seen in Table \ref{1023computation}, we compared the running time of the MCMF algorithm with that of the proposed LP method when solving the receding horizon optimization problem~\eqref{rh_opt} for a single time step with multiple prediction horizon lengths $K$. 
The LP method scales better with the prediction horizon length than the MCMF algorithm. 
When $K=16$, the LP method is approximately 880 times faster than the MCMF algorithm. 
This result demonstrates the computational efficiency of the proposed LP approach, which enables the implementation of \eqref{rh_opt} in a receding horizon fashion.

\section{Case Studies}
\label{sec:Experiment results}

\subsection{Experiment Setup}

\begin{figure}[tb]
\centering
	\includegraphics[width=3.3in]{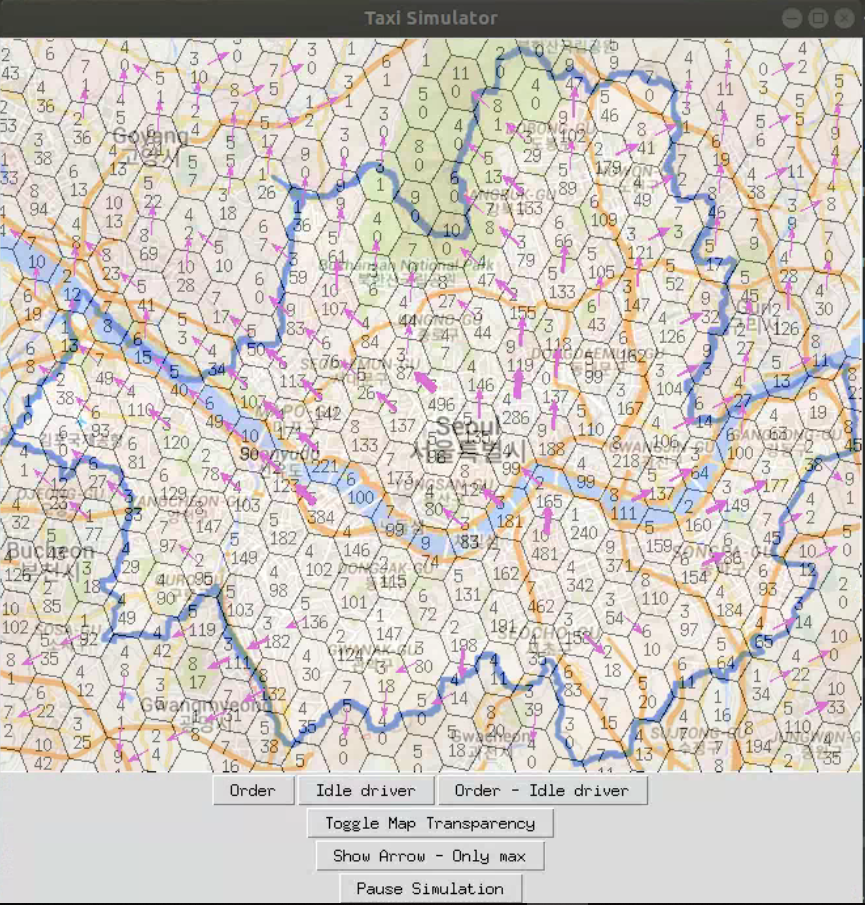}
	\caption{Simulator interface with arbitrary data. The numbers in each node represent the number of ride requests and the number of idle drivers, respectively. Each arrow indicates the vehicles' flow, and its width represents the amount of flow. 
}
	\label{simulator}
\end{figure}

We conducted case studies using the data of taxi trips in the Seoul metropolitan area in 2018.\footnote{The data were provided by Kakao Mobility,  a mobility platform company located in South Korea.}
The  data includes  information about the ride requests and the drivers' status in each cell of the hexagonal grid in two consecutive weeks in  October.
The ride request information includes the estimated fare, the estimated travel time, and the indices of the departing and  destination cells of each ride request. Moreover, the drivers' status information is used to estimate how many drivers go online and how many go offline in each cell.

In our simulator, the number of ride requests $r_{i,t}$ is given by the number of ride requests for which the departing cell is $\mathcal{N}_i$ and the starting time is $t$. 
The simulator computes the number of idle drivers $d_{i,t}$ as the number of remaining drivers in the previous time step $t-1$ plus the number of drivers going online, minus the number of drivers going offline:
\begin{align*}
d_{i,t}&:=(\mbox{\# of remaining drivers at $(t-1)$ in $\mathcal{N}_i$})\\
&\quad+(\mbox{\# of drivers becoming online in $\mathcal{N}_i$})\\
&\quad-(\mbox{\# of drivers going offline in $\mathcal{N}_i$}).
\end{align*}
When the number of drivers going offline is large, so that the left-hand side is negative, we simply define $d_{i,t}$ as 0. 
 
In addition to the realistic case where idle drivers can go online or go offline, we also consider the case where they cannot go online or go offline. In the latter case, the total number of drivers is preserved during the simulation.
A single simulation lasts 24 hours from midnight of the studied day. 
In the beginning of every single simulation, the number of idle drivers $d_{i,0}$ on each node is initialized with that of the actual data.

We generated a hexagonal grid of the map of Seoul and put an index on each grid by using H3 \cite{uberh3}.
The total number and the location of the nodes are fixed for all experiments (Figure \ref{simulator}).
There are 321 valid nodes in total, but some nodes are located in unreachable regions. We ignored these nodes when controlling the vehicles. 
All the numerical experiments were conducted on a 64-bit Ubuntu 18.04.2 LTS with Intel Core i7-8700K CPU @3.70GHz.
 The simulator was implemented by Python, and the algorithm for solving the LP problem was implemented by C and CPLEX Studio 129.

\subsection{The Oracle}
As a first comparison, we introduce the oracle, which uses all the   information about ride requests, even in the future,  to determine an optimal dispatching plan. 
The oracle also manages the movement of the fleet in three stages for each single time step, as seen in Figure \ref{fig:dispatch}. In the first stage, the oracle moves the idle drivers to neighboring nodes. During the second stage, the oracle decides whether or not each ride request is served. 
%The drivers can only take the ride requests that the oracle decides to be served. 
Finally, in the third stage, we assume that the existing vehicles are not allowed to go offline or new vehicles cannot go online. However, the vehicles that complete the ride request can go online at their destination cells. 

The oracle's dispatch solution  maximizes the total profit from the ride-hailing service throughout the day. 
To set up the exact optimization problem for the oracle, let $m$ be the total number of the requests during a day. 
We denote the $a$th request as a 5-tuple  $R^a=(R^a_{\textup{dep}},R^a_{\textup{dest}},R^a_{\textup{start}},R^a_{\textup{dur}},R^a_{\textup{price}})$, where each component of $R^a$ has the following information:
\begin{itemize}
	\item $R^a_{\textup{dep}}$ : the grid index of departing node of the request
	\item $R^a_{\textup{dest}}$: the grid index of destination node of the request
	\item $R^a_{\textup{start}}$: the start time of the request
	\item $R^a_{\textup{end}}$: the end time of the request
	\item $R^a_{\textup{price}}$: the price of the request.
\end{itemize}

Moreover, we define $\mathcal{R}_{i,k}^{\textup{dep}}$ as the set of requests whose departing node and time are $\mathcal{N}_i$ and $k$, respectively: 
\[\mathcal{R}_{i,k}^{\textup{dep}}:=\left\{a\mid R_{\textup{dep}}^a=i,\; R_{\textup{start}}^a=k\right\}.\]
Similarly, we define $\mathcal{R}_{i,k}^{\textup{dest}}$ as the set of requests whose destination node and time are $\mathcal{N}_i$ and $k$, respectively:
\[\mathcal{R}_{i,k}^{\textup{dest}}:=\left\{a\mid R_{\textup{dest}}^a=i,\; R_{\textup{end}}^a=k\right\}.\]
Note that the cardinality of $\mathcal{R}_{i,k}^{\textup{dep}}$ is $r_{i,k}$, the number of requests in cell $\mathcal{N}_i$ and in time step $k$. 

Finally, since the oracle has to decide whether or not each request is served, we introduce the decision variables whose values are 0 or 1. Precisely, for each $(i,k)$, we introduce the vector $b_{i,k}=(b^{a_1}_{i,k},\ldots,b^{a_{r_{i,k}}}_{i,k})\in \{0,1\}^{r_{i,k}}$ denoting the decision of the oracle on the requests in $\mathcal{R}_{i,k}^{\textup{dep}}$, where $a_1,\ldots, a_{r_{i,k}}\in \mathcal{R}_{i,k}^{\textup{dep}}$. Here, $b_{i,k}^{a}=1$ implies that the oracle decides to serve the request $a\in\mathcal{R}_{i,k}^{\textup{dep}}$. We denote the entire vector of $b_{i,k}$ as  $b=(b_{1,0},b_{2,0},\ldots,b_{n-1,T-1}, b_{n,T-1})\in \{0,1\}^m$.

Given all the data of ride requests $(R^1,\ldots, R^m)$, the optimization problem for the oracle can be formulated as the following mixed-integer  linear program:
\begin{align}
\begin{aligned} \label{oracle}
\max_{b,x}\quad&\sum_{i=1}^n\sum_{k=0}^{T-1}\sum_{a\in \mathcal{R}_{i,k}^{\textup{dep}}} b_{i,k}^a R_{\textup{price}}^a -\sum_{i,j,k} c_{i\to j,k}x_{i\to j,k},\\
\mbox{s.t.}\quad & \sum_{a\in \mathcal{R}_{i,t}^{\textup{dep}}}b_{i,t}^a+s_{i,t}=d_{i,t}+\sum_{j=1}^n x_{j\to i,t}-\sum_{l=1}^nx_{i\to l,t} \\
&d_{i,t}=s_{i,t-1}+\sum_{j=1}^n\sum_{k=0}^{t-1} \sum_{a\in \mathcal{R}_{j,k}^\textup{dep}\cap \mathcal{R}_{i,t}^\textup{dest}} b_{j,k}^a \\
&\sum_{j=1}^n x_{i\to j,t}\le d_{i,t},\quad x_{i\to j,t}\ge 0,\quad s_{i,t}\ge0,
\end{aligned}
\end{align}	
where all the constraints are satisfied for all $t=0,1,\ldots, T-1$ and $i=1,2,\ldots, n$, except the second constraint which holds for $t\ge 1$.
The first constraint  represents the balance law of taxi flow. More precisely, $\sum_{a}b_{i,t}^a$ implies the number of served requests in node $\mathcal{N}_i$ in time step $t$, and $s_{i,t}$ denotes the number of remaining drivers in node $\mathcal{N}_i$ after the requests in  time step $t$ are served. This quantity should be the same as the number of drivers before the requests are served, which is the right-hand side of the first constraint. 
As specified in the second constraint, the number of idle drivers $d_{i,t}$ should be the same as the sum of the number of remaining drivers $s_{i,t-1}$ and the number of drivers going online after serving the requests, whose destination and end time are $\mathcal{N}_i$ and $t$, respectively.

To solve \eqref{oracle}, we relax the integer constraint $b\in\{0,1\}^m$ to the constraint $0\le b_{i,t}^a\le 1$. 
Then, the relaxed problem is an LP. 
However, the optimal $b$ of the relaxed LP, obtained by CPLEX,
 is an integer vector in all of our numerical experiments.\footnote{In fact, given the relaxed LP of \eqref{oracle}, it is possible to construct the equivalent flow problem with integer capacity. Then, Theorem \ref{T1} guarantees that $b$ is an integer vector.}
Thus, this solution is optimal for the original problem~\eqref{oracle}. 
 Note that the performance of the oracle has to be better than that of any other algorithm, since the oracle uses all of the ride request data and foresees the future. For example, the oracle can handle  unpredictable ride requests caused by special events, while other algorithms cannot.

\begin{figure}[tb] 
	 \centering
    \begin{subfigure}{0.49\textwidth}
    	 \centering
	\includegraphics[width=3in]{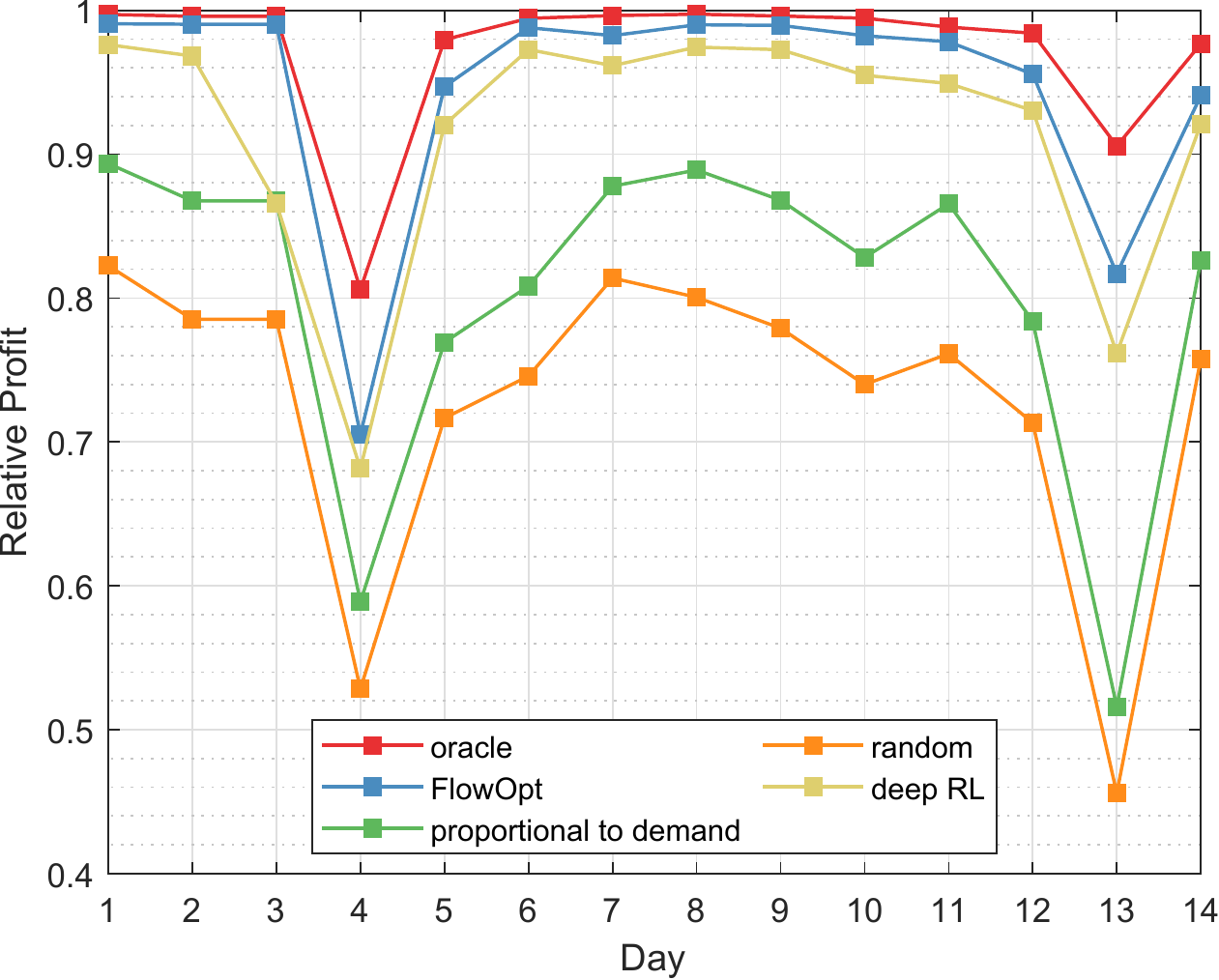}	
	  \caption{}
    \end{subfigure}
    \begin{subfigure}{0.49\textwidth}
    	 \centering
	\includegraphics[width=3in]{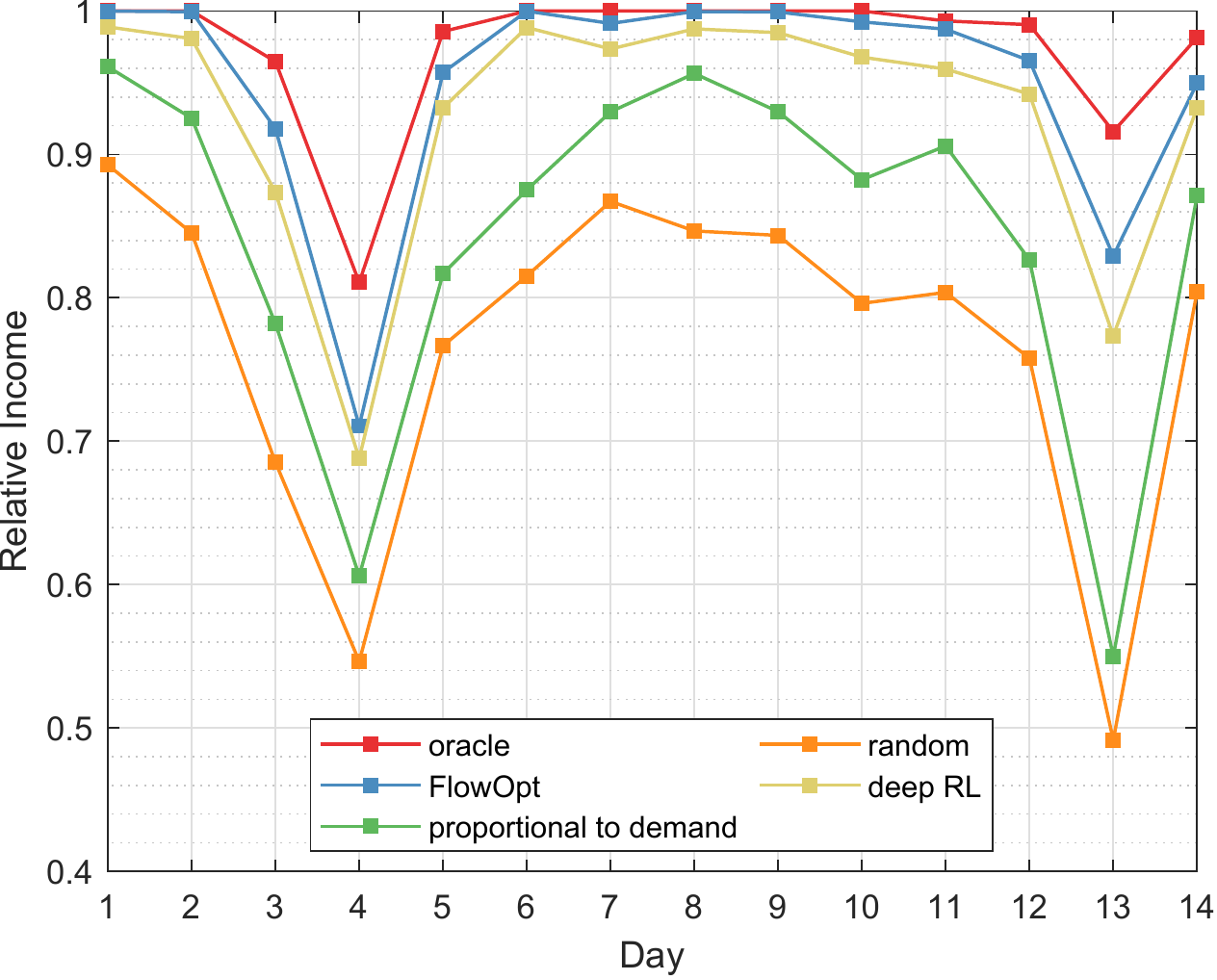}
	  \caption{}
    \end{subfigure}
	\caption{Performance of the algorithms without allowing the drivers to go offline when not serving a ride or new drivers to be online: (a) relative profit, and (b) relative income. }\label{figure1}
\end{figure}

\subsection{Other algorithms}

To compare our method  with other algorithms, we considered a deep reinforcement learning (RL)  algorithm, the random-move algorithm, and the proportional-to-demand algorithm. The deep RL algorithm is an actor-critic method in which drivers in the particular grid cell use  information  within the radius of three hexagons~\cite{linkaixiang2018efficient, Kim2019}. 
The random-move algorithm equally distributes the idle drivers at a particular node to the adjacent nodes. 
The proportional-to-demand algorithm distributes the idle drivers at a particular node to the adjacent nodes, proportional to the number of requests in each node. We provide the details of these algorithms in Appendix.

\subsection{Comparison}

\subsubsection{Results with the original driver data}

We compared the performance of the aforementioned algorithms with that of the proposed flow network-based LP method, which we refer to as {\it FlowOpt} throughout this subsection. 
We first used the original initial driver data $d_{i,0}$; simulations with the modified initial driver data will be discussed in the following subsection. We set the prediction horizon $K=30$ for FlowOpt, since 30 is slightly larger than the diameter of the directed graph generated by the map of Seoul. Therefore, any vehicle starting at a particular node can reach all the other nodes during the prediction horizon of FlowOpt.

We used two performance measures. The first is the {\it relative profit}, given by
\begin{align*}
\frac{(\mbox{Actual GMV}) - (\mbox{Total reposition cost})}{(\mbox{Maximal GMV})}.
\end{align*}
The second is the {\it relative income}, given by
\begin{align*}
\frac{(\mbox{Actual GMV})}{(\mbox{Maximal GMV} )},
\end{align*}
where the gross merchandise volume (GMV) represents the sum of served requests' fares, as discussed in the beginning of  Section~\ref{sec:setup}.
The maximal GMV is achieved if all the ride requests are served.

\begin{table}[tb]
\normalsize
\centering
\begin{tabular}{ l  || C{0.9in} | C{0.9in}} 
\hline
		Algorithm		& Relative profit	 & Relative  income  \\ \hline\hline
{Oracle} 	&  0.9720 			& 0.9744 \\ 
{\bf FlowOpt}  &   {\bf 0.9462} 			& {\bf 0.9499}   \\ 
{Deep RL} &	0.9150            & 0.9267 \\ 
{Prop-to-demand} & 0.8035            & 0.8441\\
{Random-move} &0.7289            & 0.7687 \\ 
\hline
\end{tabular}
\caption{
The relative profit and income (on average) obtained by the algorithms without allowing the drivers to go offline when not serving a ride or new drivers to be online.
}\label{tab:perf1}
\end{table}

%{\color{red}
%	\begin{table}[tb]
%		\normalsize
%		\centering
%		\begin{tabular}{ l  || C{0.9in}}\hline
%			& computational time(seconds)  \\ \hline\hline
%			mean  & 2.372\\ 
%			max &  3.162\\ 
%			stdev & 0.7746          \\
%			\hline
%		\end{tabular}
%		\caption{
%			{\color{red}The representative values of computational time data of 144 time steps in day 13 without allowing drivers to go online/offline. }
%		}\label{time}
%	\end{table}
%}

Figure \ref{figure1}  displays the performance of FlowOpt, the oracle, the deep RL algorithm, and the two rule-based algorithms without allowing the drivers to go offline when not serving a ride or new drivers to be online.
The  result indicates that  FlowOpt performs almost as good as the oracle, and it significantly outperforms the two rule-based algorithms. 
This is a surprising result, because the oracle uses all the data and  foresees the future, while  FlowOpt only uses  the data of the current time step, together with the receding horizon approximation. 
This implies that the receding horizon approximation not only simplifies the problem, it  also provides a reasonable approximation. 
One interesting feature shown in Figure \ref{figure1} is that the performance of all the methods is degraded on Day 4 and Day 13. 
These unexpected sudden drops in performance are due to special events held in Seoul on those days. 
Table~\ref{tab:perf1} summarizes the relative profit and income averaged over a period of  two weeks. 
On average, the performance of FlowOpt was found to be  97.16\% of the oracle's performance. Furthermore, we note that the computational time of FlowOpt for determining a taxi dispatching plan for a single time step is around 3 seconds in maximum. Since the time interval of the single time step is 10 minutes, the computational time for FlowOpt is sufficiently small and therefore, FlowOpt can be used in real-time taxi dispatching.

\begin{figure}[tb]
	  \centering
    \begin{subfigure}{0.49\textwidth}
        \includegraphics[width=3in]{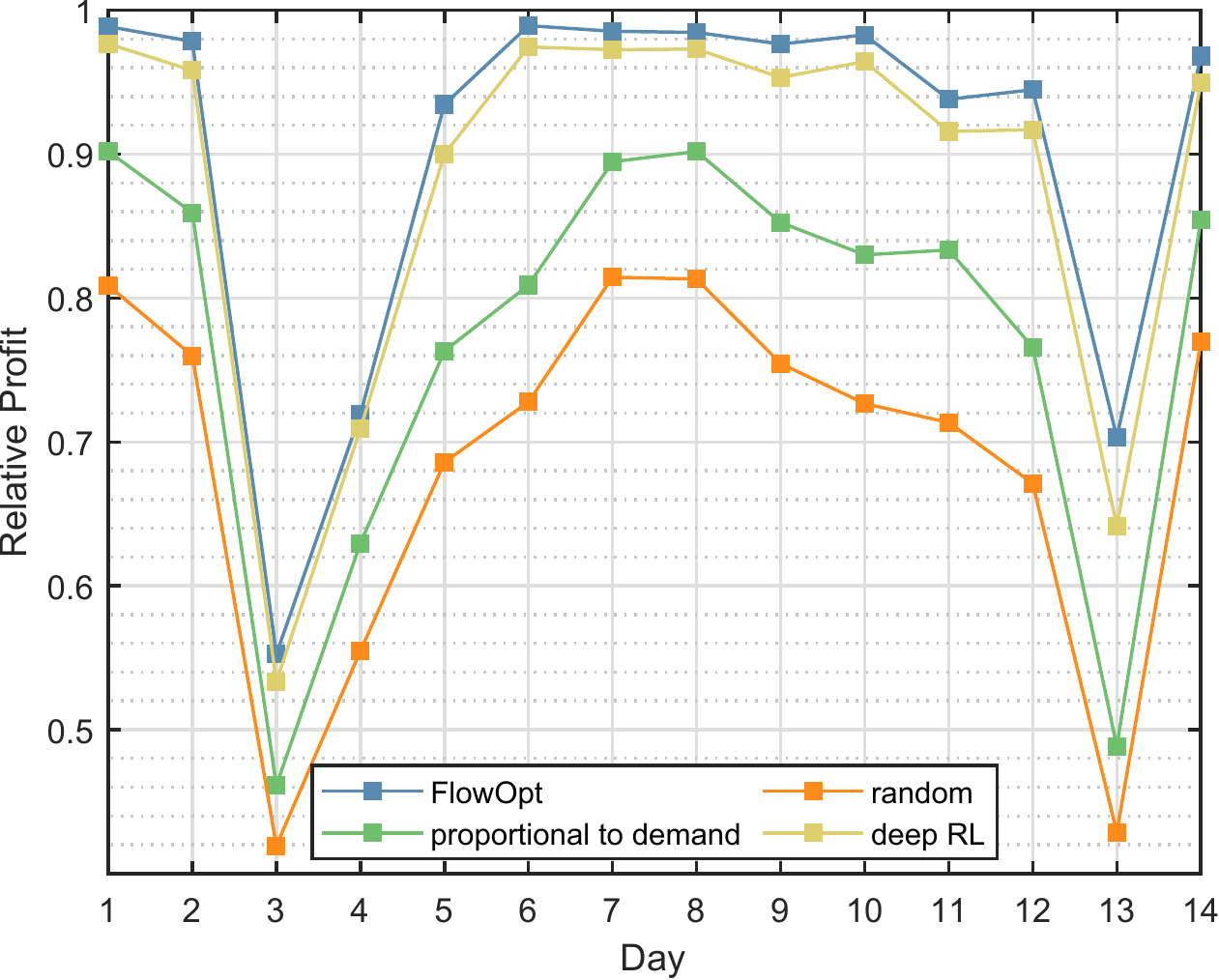}
		  \caption{}
    \end{subfigure}
    \begin{subfigure}{0.49\textwidth}
        \includegraphics[width=3in]{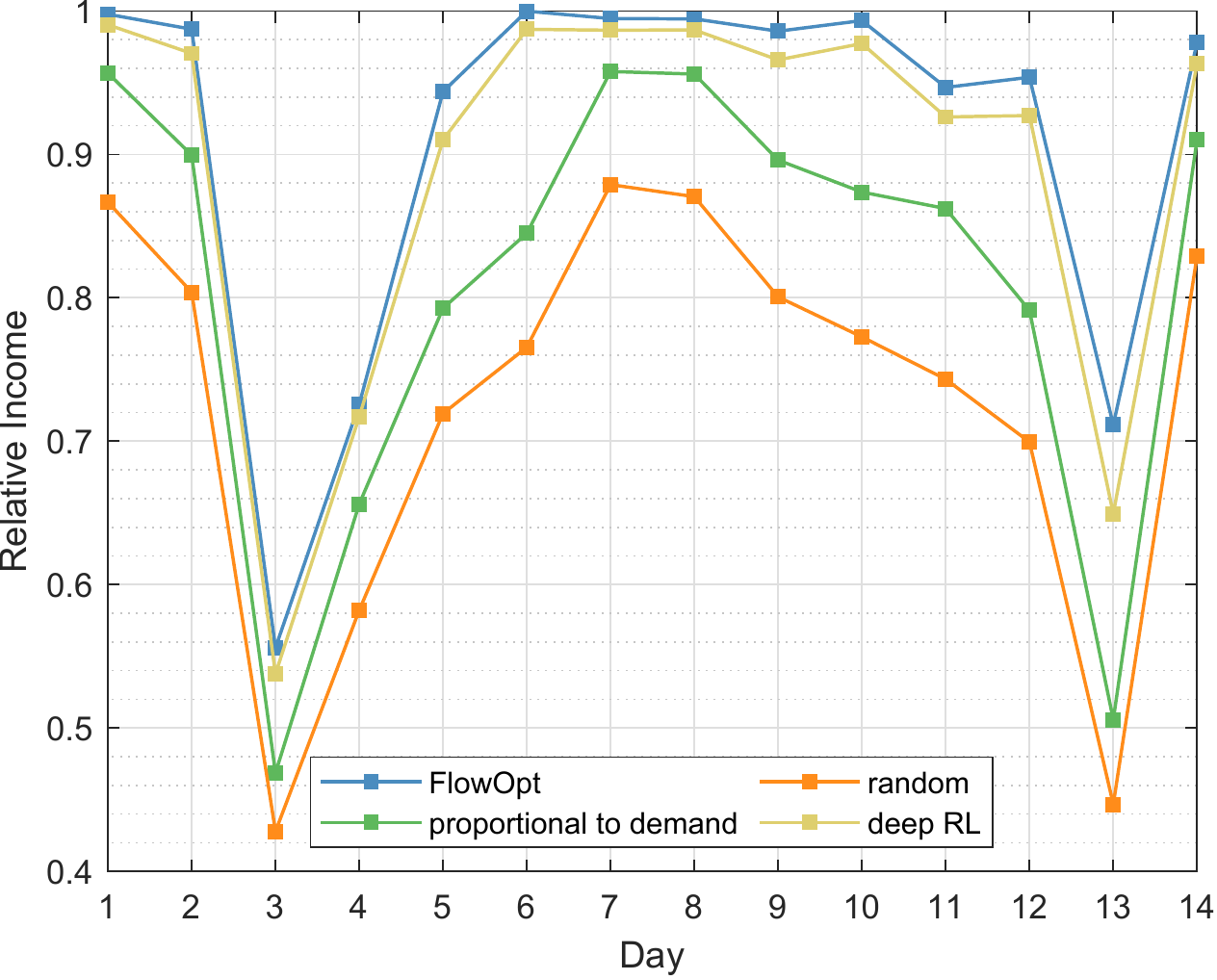}
	  \caption{}
    \end{subfigure}
	\caption{Performance of the algorithms when allowing drivers to go online/offline: (a) relative profit, and (b) relative income. }
	\label{figure2}
\end{figure}

Figure \ref{figure2}  shows the performance of all the algorithms, except the oracle,  when allowing drivers to go online and go offline.  
We excluded the oracle in this simulation experiment, since it results in  an undesirable outcome. 
Because the oracle uses the status data of real drivers, it knows in which hexagonal cell the real drivers go offline. 
Thus, the oracle tends not to move the idle drivers    to the cells where the real drivers go offline.
 This prevents drivers  from going offline; therefore, whether the drivers go online and offline  becomes meaningless to the oracle.
 Again, FlowOpt  demonstrated the best performance, although the overall performance was decreased by approximately 5\% from the previous case. We note that the performance on Day 3 is significantly lower than it was in the previous case.
According to the driver data on Day 3, a large number of drivers go offline in the early morning (around 5AM). 
Thus, when allowing drivers to go online/offline,
there is a shortage of drivers during the morning commutes, so the performance decreases. 
In Table~\ref{tab:perf2}, we report the relative profit and income averaged over the two-week period when allowing drivers to go online/offline.

\begin{table}[tb]
\normalsize
\centering
\begin{tabular}{ l  || C{0.9in} | C{0.9in}}\hline
	Algorithm			& Relative profit	 & Relative  income  \\ \hline\hline
{\bf FlowOpt}  & {\bf 0.9033}            & {\bf 0.9121}  \\ 
{Deep RL} &  0.8814            & 0.8925 \\ 
{Prop-to-demand} & 0.7746            & 0.8122\\
{Random-move} &0.6891            & 0.7290 \\ 
\hline
\end{tabular}
\caption{
The relative profit and income (on average) obtained by the algorithms when allowing drivers to go online/offline. 
}\label{tab:perf2}
\end{table}

\subsubsection{Results with reduced driver data}
In the data, there are a sufficient number of taxi drivers to serve almost all ride requests. 
Thus, we tested if FlowOpt can perform as good as the oracle, even when there are not enough drivers. 
We modulated the initial number of idle drivers $d_{i,0}$ by multiplying the ``driver multiplier" so that  the algorithms were evaluated using fewer initial drivers.
Figure \ref{figure3} displays the performances of the algorithms with fewer drivers on Day 1 without allowing the drivers to go offline when not serving a ride or new drivers to be online. Obviously, as the number of drivers decreases,  it is less likely that an arbitrary ride request is served. 
This leads to a decrease in the GMV. 
This explains why the relative profit reduces as the driver multiplier getting smaller. 
Nevertheless, FlowOpt  maintains the best performance for every choice of driver multiplier among the other algorithms,  except the oracle.

\begin{figure}[tb]
	  \centering
        \includegraphics[width=3.35in]{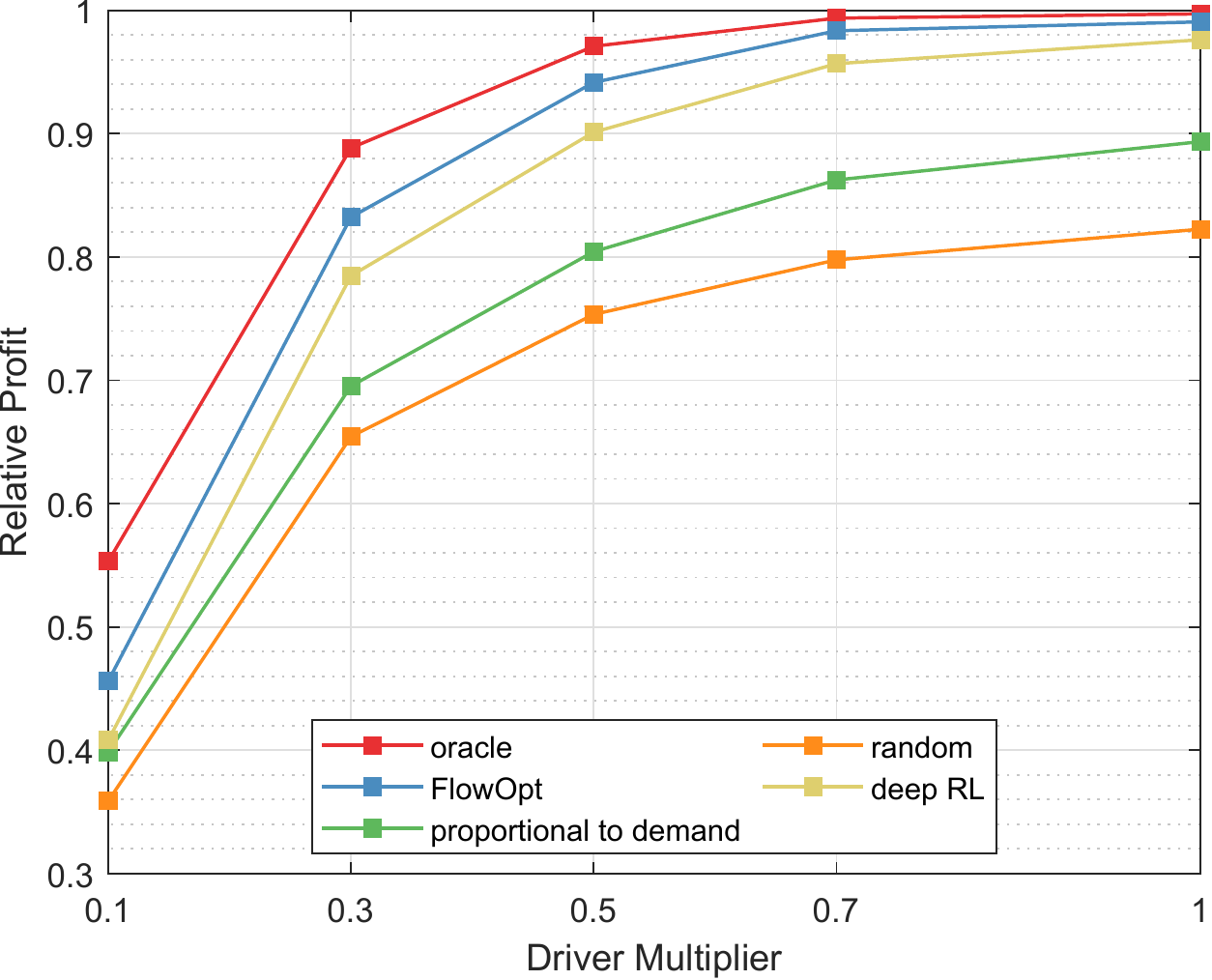}
	\caption{Relative profit of the algorithms with fewer drivers on Day 1.}
\label{figure3}
\end{figure}

\subsubsection{Reposition costs}
We then considered the reposition costs incurred by the algorithms. 
Figure \ref{figure5} shows the ratio between the total reposition cost to  the maximal GMV for each algorithm without allowing the drivers to go offline when not serving a ride or new drivers to be online.
\begin{figure}[tb]
\centering
	\includegraphics[width=3.35in]{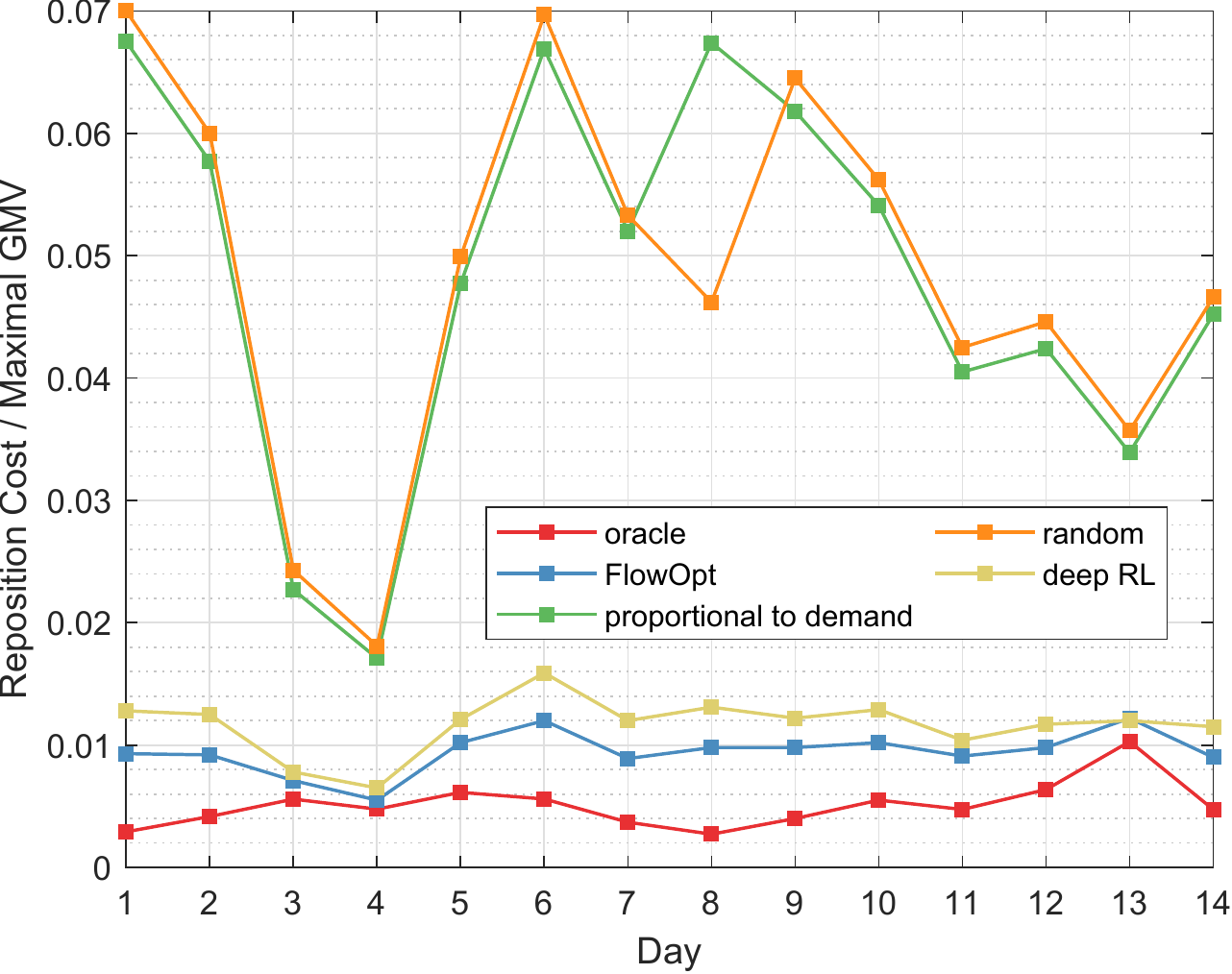}
	\caption{The ratio between the total reposition cost to  the maximal GMV.}
	\label{figure5}
\end{figure}
We observed that  FlowOpt, the oracle and the deep RL algorithm maintain a reposition cost that is almost negligible. 
The oracle has the lowest reposition cost, and  FlowOpt and the deep RL  algorithm have  reposition costs that are similar to that of the oracle.
The two rule-based algorithms have significantly higher reposition costs than FlowOpt since they do not have any explicit mechanism to reduce the reposition cost. Moreover, the large fluctuations of two rule-based algorithms are also due to the absence of a mechanism to control the reposition cost. Since those algorithms just send idle vehicles to adjacent nodes without considering the reposition cost, they cannot control the reposition cost. Therefore, the ratio between reposition cost and maximal GMV severely depends on the dynamics of maximal GMV. On the other hand, the other algorithms can control the reposition cost by considering it as a part of an objective function of the optimization problem. Thus, although the data changes day to day, those algorithms can stably maintain the ratio between the reposition cost and maximal GMV.

\subsection{Effect of the Prediction Horizon} 

\begin{figure}[tb]
	\centering
        \includegraphics[width=3.35in]{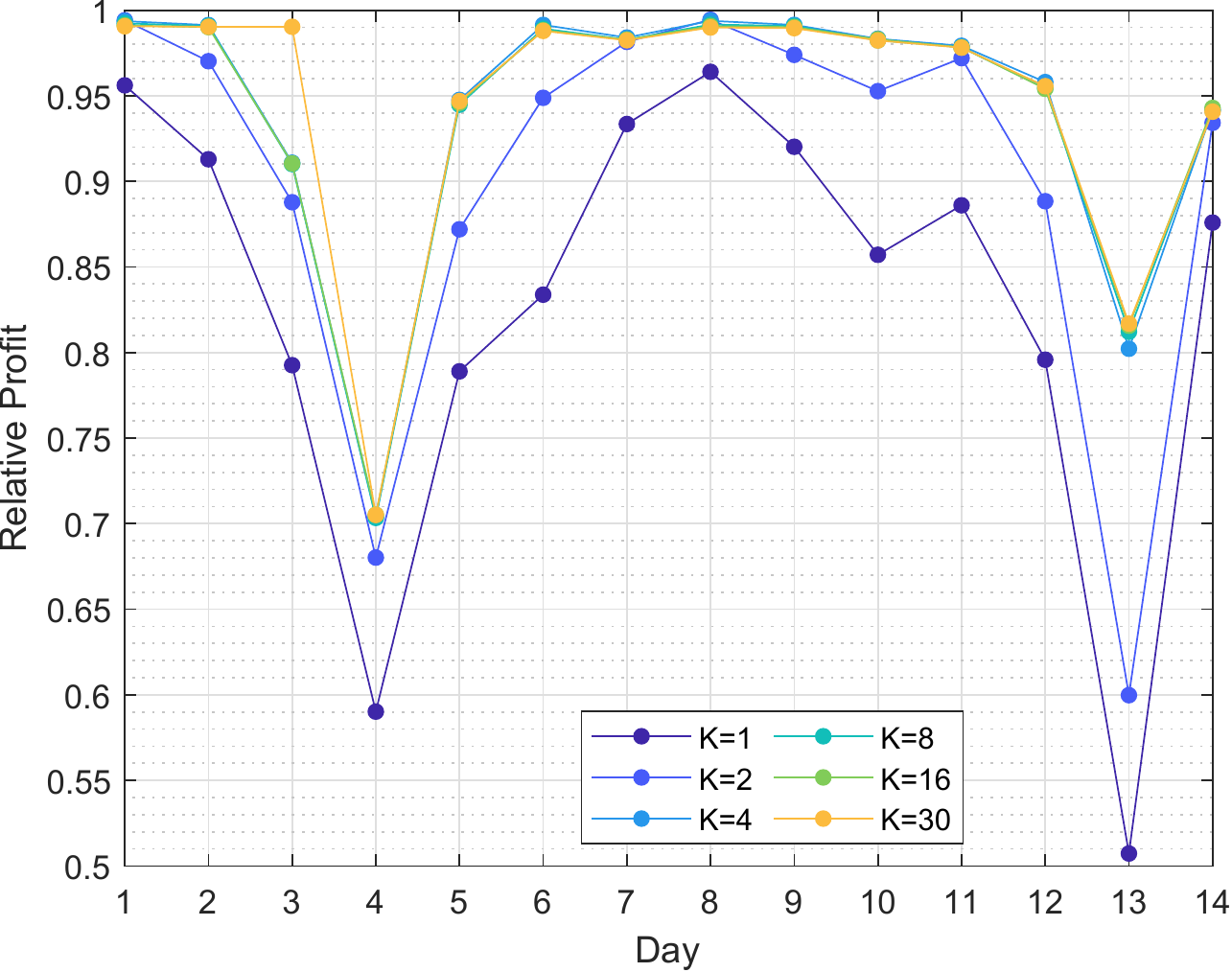}
%	\vspace{-5pt}
	\caption{Effect of the prediction horizon length on the relative profit of FlowOpt.}
\label{figure6}
\end{figure}

Lastly, we examined the effect of the prediction horizon length $K$ on the performance of FlowOpt. 
Figure \ref{figure6}    displays the performance with various prediction horizon lengths $K$ without allowing the drivers to go offline when not serving a ride or new drivers to be online.
The performance tends to increase with the prediction horizon length. 
When $K$ is small, FlowOpt   works well for ordinary days, including Day 1, Day 2, Day 9 and Day 14. 
However, on days when there are special events in Seoul, such as Day 3, Day 4, and Day 13, 
there is a significant benefit of using a large $K$.

\section{Conclusion}\label{sec:conclusion}

In this paper, we proposed  a  predictive taxi dispatch method  to serve as many ride requests as possible while minimizing the cost of repositioning vehicles.
By converting the multi-objective receding horizon optimization problem into a MCMF problem, we obtained two desired theoretical properties:   the integer property and the time-greedy property. 
For a  near real-time dispatch solution, an LP method was also developed. 
According to the experiment results using the data for Seoul, South Korea, 
the performance of the proposed MPC method is almost as good as  that of the oracle, and it outperforms the other methods that were evaluated. 
The performance of our solution also increases with the prediction horizon length. 

The proposed taxi dispatch method can be further extended in many interesting ways. First, it can be used in conjunction with statistical learning methods that predict  demand. 
Second, a robust version of the proposed MPC method can  be considered to overcome the issues that arise from prediction errors and uncertainties.
Third, by using real-time traffic information,
the travel time of vehicles can also be considered in the proposed multi-objective setting.

\appendix
\section{Explanation of the other algorithms used in the simulations}
\subsection{Deep RL algorithm}
This subsection provides details about the deep RL algorithm used in our simulation studies. 
The RL algorithm uses  the advantage actor-critic method 
that consists of one actor and one critic, both of which are parameterized by deep neural networks~\cite{mnih2016asynchronous}.
This actor-critic pair is optimized to yield the profit-maximizing dispatch solution for an arbitrary node.
To be specific, our RL agent provides a generic solution to the distribution of requests and drivers, without considering the spatio-temporal information about a node (e.g., time, node ID, or geographic coordinate).
Given a node index $i\in\{1,\dots,N\}$, the goal of the optimization is to maximize the discounted sum of the net incomes earned in $\mathcal{N}_i$ throughout all time steps
\[
\sum_{t=0}^{\infty} \sum_{a\in \mathcal{R}_{i,t}^{\textup{dep}}} \gamma^{t} \bigg ( b_{i,t}^a R_{\textup{price}}^a -\sum_{j=1}^{n} c_{i\to j,t}x_{i\to j,t} \bigg ),
\]
where $\gamma$ is the discount factor set to be $0.9$.
Note that the original problem is to optimize the net income for one day, but the RL agent solves the infinite horizon problem because it neglects the time limit.
The RL agent assumes that the simulation runs forever, and the distribution of ride requests is the same for each day.

The state is given by a vector that includes the number of requests $r_{\cdot,t}$ and the number of drivers $d_{\cdot,t}$ in a node and its neighboring nodes.
We define a neighboring node as a node within the first, second, and third layer in a hexagonal grid~\cite{Kim2019}.
In other words, no more than two hexagons are in between the given node to a neighboring node.
Thus, the dimensionality of a state is $2\times(1+6+12+18) = 74$.
The action is a $7$-dimensional vector where each entry indicates the number of drivers to move towards the specific direction (including staying in the current node).
Ideally, an actor should provide a dispatch solution for a fleet in the node; however, it is difficult to implement such an actor in a neural network due to the high dimensionality of the action space.
Thus we construct the actor that computes a policy for a single driver.
That is, the actor of our RL agent yields the {\it probability} of choosing the directions.
If some directions are inaccessible, the probability of reaching them is set to $0$, then the probability vector is normalized to $1$.
Finally, each driver selects its direction by sampling from the probability vector.

The actor and the critic are implemented in fully connected networks.
Each network consists of three hidden layers and an output layer.
The hidden layers have output dimensions of $64$, $32$, and $16$, from the first layer to the third layer, respectively, but the actor and the critic do not share any weights.
For all layers, including the output layer, we use ReLU as the activation function.

We construct multiple RL agents, as described above, and  optimize each agent for each day of the ride-hailing service.
Each agent repeats the learning episode $100$ times.
In each episode, we run the simulation once using the actor, store the simulation data related to learning in the memory, and train the actor-critic with the memory.
The memory contains the number of drivers, the number of ride requests, and the revenue along with the reposition cost in each node.
The information is kept intact throughout an episode, but it is removed from the memory after an episode is finished.
For each episode, mini-batch gradient descent is applied to the actor-critic, $4000$ times.
A mini-batch has the size of $3000$ samples.
We use Adam  \cite{kingma2015adam} to minimize the loss function in training the neural networks.  
The learning rates for the actor and the critic are $0.001$ and $0.005$, respectively.

\subsection{Proportional-to-demand and Random-move algorithms}
This subsection describes  the two rule-based algorithms, namely the proportional-to-demand and the random-move algorithms. For any grid cell $\mathcal{N}_i$, let $\mathcal{N}_{i_1}, \ldots ,\mathcal{N}_{i_l}$ be its neighboring hexagonal cells. 
In each time step $t$, the proportional-to-demand algorithm sends $n_{i_j}$ drivers from $\mathcal{N}_i$ to $\mathcal{N}_{i_j}$, where $n_{i_j}$ is proportional to the number of requests in $\mathcal{N}_{i_j}$. Precisely, the proportional-to-demand algorithm sends $\lfloor\omega_{i,j}d_{i,t}\rfloor$ drivers to the adjacent cell $\mathcal{N}_{i_j}$, where the weight $\omega_{i,j}$ is given by
\[\omega_{i,j}:=\frac{r_{i_j,t}}{r_{i,t}+\sum_{k=1}^l r_{i_k,t}},\]
and $\lfloor x\rfloor$ is the largest integer that does not exceed $x$. The remaining $d_{i,t}-\sum_{j=1}^{l}\lfloor \omega_{i,j}d_{i,t}\rfloor$ drivers stay in $\mathcal{N}_i$.

On the other hand,  the random-move algorithm  does not consider the number of requests in the neighboring cells. It sends $\lfloor\omega_{i,j}d_{i,t}\rfloor$ drivers from $\mathcal{N}_i$ to the adjacent cell $\mathcal{N}_{i_j}$, where the weight is chosen as $\omega_{i,j}=\frac{1}{l+1}$.

\bibliographystyle{IEEEtran}
\bibliography{related_work,refered_methods}

\end{document}